\documentclass[12pt, a4paper]{article}
\usepackage[top=1in,bottom=1in,left=1in,right=1in]{geometry}
\usepackage{epsfig, amsmath, amssymb, latexsym, graphicx, setspace, sectsty}
\usepackage{multirow, enumerate, url, natbib, booktabs, xcolor, bbm, amsthm, algorithm}
\usepackage{authblk}
\usepackage{float} 
\floatstyle{plain}
\newfloat{Algorithm}{thp}{lop}
\floatname{Algorithm}{Algorithm}

\RequirePackage[colorlinks,citecolor=blue,urlcolor=blue]{hyperref}
\sectionfont{\large}
\subsectionfont{\normalsize}

\newtheorem{theorem}{Theorem}

\newcommand{\argmax}{\operatornamewithlimits{argmax}} 
 
\def\e{{\rm e}}   
\def\diag{\text{\rm diag}}

\def\tr{\text{\rm tr}}

\def\cov{\text{\rm cov}}
\def\vec{\text{vec}}
\def\vech{\text{vech}}
\def\IW{\text{\rm IW}}
\def\hnabla{\widehat{\nabla}}

\def\ESS{\text{\rm ESS}}
\def\gwd{\text{\rm gwd}}
\def\gwesp{\text{\rm gwesp}}
\newcommand\ML{\text{ML}}
\newcommand\PL{\text{PL}}
\newcommand\BF{\text{BF}}
\def\vijr{{\hat{v}_{ij}^{(r)}}}
\def\mijr{{\hat{m}_{ij}^{(r)}}}
\def\L{{\mathcal{L}}}
\def\M{{\mathcal{M}}}
\def\N{{\mathcal{N}}}
\def\D{{\mathcal{D}}}

\def\V{{\mathcal{V}}}
\def\H{{\mathcal{H}}}

\graphicspath{{img/}}

\title{Bayesian variational inference for \\[-4mm]
exponential random graph models}

\author[1]{Linda S. L. Tan}
\author[2]{Nial Friel}
\affil[1]{National University of Singapore}
\affil[2]{University College Dublin}
\date{}

\begin{document}
\maketitle
\vspace{-10mm}
\begin{abstract}
Deriving Bayesian inference for exponential random graph models (ERGMs) is a challenging ``doubly intractable" problem as the normalizing constants of the likelihood and posterior density are both intractable. Markov chain Monte Carlo (MCMC) methods which yield Bayesian inference for ERGMs, such as the exchange algorithm, are asymptotically exact but computationally intensive, as a network has to be drawn from the likelihood at every step using, for instance, a ``tie no tie" sampler. In this article, we develop a variety of variational methods for Gaussian approximation of the posterior density and model selection. These include nonconjugate variational message passing based on an adjusted pseudolikelihood and stochastic variational inference. To overcome the computational hurdle of drawing a network from the likelihood at each iteration, we propose stochastic gradient ascent with biased but consistent gradient estimates computed using adaptive self-normalized importance sampling. These methods provide attractive fast alternatives to MCMC for posterior approximation. We illustrate the variational methods using real networks and compare their accuracy with results obtained via MCMC and Laplace approximation.

\end{abstract}

\section{Introduction} \label{sec:Intro}
Exponential random graph models (ERGMs) are widely used in economics, sociology, political science and public health to analyze networks. Fitting ERGMs using maximum likelihood estimation (MLE) is challenging as the normalizing constant of the likelihood involves a sum over all possible networks, which is intractable except for very small networks. Bayesian inference for ERGMs is even more challenging as the normalizing constant of the posterior density is also intractable, leading to a {\em doubly intractable} problem. Recently, a number of Markov chain Monte Carlo (MCMC) methods have been developed to address this problem. These include auxiliary variable approaches \citep{Moller2006, Murray2006} such as the double Metropolis-Hastings sampler \citep{Liang2010} and adaptive exchange algorithm \citep{Liang2016}, and likelihood approximation methods \citep{Atchade2013} such as the Russian roulette algorithm \citep{Lyne2015}. \cite{Park2018} provide a comprehensive review of these techniques. As MCMC methods are computationally intensive, we propose fast variational methods as alternatives for obtaining Bayesian inference for ERGMs. 

The classic approach for fitting ERGMs is MCMC MLE, which is first described by \cite{Geyer1992} and developed for ERGMs by \cite{Snijders2002}. To overcome the intractability of the log likelihood $\ell(\theta)$ for parameters $\theta$, MCMC MLE maximizes an estimate of $\ell(\theta) - \ell(\theta_0)$, where $\theta_0$ is a fixed value that should ideally be close to the maximum likelihood estimate $\hat{\theta}_\ML$. Success of this method rests crucially on the choice of $\theta_0$, and a poor choice may result in an objective function that cannot be maximized \citep{Caimo2011}.  \cite{Hummel2012} introduce a method to move $\theta_0$ closer to $\hat{\theta}_\ML$ sequentially. An alternative is maximum pseudolikelihood estimation \citep[MPLE,][]{Besag1974}, where the likelihood is approximated by a product of full conditional distributions, assuming that the dyads are conditionally independent given the rest of the network. While MPLE is fast, it can result in unreliable inference. 

To derive Bayesian inference for ERGMs, \cite{Caimo2011} propose an MCMC algorithm that samples from the likelihood using a ``tie no tie" sampler \citep{Hunter2008} and draws posterior samples of $\theta$ using the exchange algorithm \citep{Murray2006}. \cite{Caimo2013} extend this algorithm into a reversible jump MCMC algorithm. Sampling from the likelihood is generally a time-consuming procedure, whose convergence is difficult to assess \citep{Bouranis2017}. This creates obstacles in model selection, where many candidate models are to be fitted in a short time. \cite{Bouranis2018} propose an affine transformation of $\theta$ for correcting the mode, curvature and magnitude of the pseudolikelihood. When used in MCMC methods, this {\em adjusted pseudolikelihood} yield Bayesian inference for ERGMs at a much lower computational cost as sampling from the likelihood at each iteration is no longer necessary. 

In this article, we develop a variety of variational methods for Gaussian posterior approximation for ERGMs. First, a nonconjugate variational message passing (NCVMP) algorithm \citep{Knowles2011} is developed using the adjusted pseudolikelihood. This algorithm converges rapidly, but the accuracy of the posterior approximation is tied to how well the adjusted pseudolikelihood mimics the true likelihood. We consider an alternative stochastic gradient ascent algorithm \citep{Titsias2014}, which uses a reparametrization trick \citep{Kingma2014, Rezende2014} to transform variables so that the gradients are direct functions of the variational parameters. This method does not rely on the adjusted pseudolikelihood. However, gradient estimation poses a challenge as the likelihood is intractable, and forming an unbiased gradient requires drawing networks from the likelihood. We explore two solutions. The first uses Monte Carlo sampling, where only a very small number of networks are drawn from the likelihood at each iteration. This approach is feasible in stochastic gradient ascent as convergence is ensured with unbiased gradients and appropriate stepsize. 

The second approach alleviates the burden of sampling at each iteration by using biased but consistent gradients computed using self-normalized importance sampling (SNIS). We propose a novel adaptive sampling strategy, where a set $\mathbb{S}$ of particles and sufficient statistics (of networks drawn from the likelihood of these particles) is maintained at any iteration. Given $\theta$, the gradient is first computed using SNIS by using the sufficient statistics of the particle closest to $\theta$ in $\mathbb{S}$. If the SNIS estimate is poor, a set of networks are sampled from the likelihood of $\theta$ and the gradient is estimated using Monte Carlo instead. Subsequently, $\theta$ and its sufficient statistics are added to $\mathbb{S}$. We begin with one particle ($\hat{\theta}_\ML$) and allow $\mathbb{S}$ to grow as the{ algorithm proceeds. This strategy reduces computation for low-dimensional problems, albeit with increased storage. 

Variational methods seek to minimize the Kullback-Leibler (KL) divergence between the true posterior and variational density, which is equivalent to maximizing a lower bound $\mathcal{L}$ on the log marginal likelihood. \cite{Tran2017} extend variational Bayes to intractable likelihood problems by showing that an unbiased gradient of the KL divergence can be obtained by replacing the likelihood with an unbiased estimate. Here, we do not estimate the likelihood explicitly or assume it is easy to simulate from the likelihood. Instead, we apply the reparametrization trick before showing that an unbiased gradient of $\mathcal{L}$ can be obtained by simulating from the likelihood. 

While the log marginal likelihood is useful in model selection, using $\mathcal{L}$ as a substitute may not yield reliable results as the tightness of the bound is unknown. \cite{Burda2016} propose an importance weighted lower bound $\L_V^\IW$, which can be computed by generating $V$ samples from the variational density. $\L_V^\IW$ increases monotonically with $V$ and approaches the log marginal likelihood in the limit. As $\L_V^\IW$ is an asymptotically unbiased estimator of the log marginal likelihood, it can be useful for model selection for sufficiently large $V$. We investigate the accuracy, efficiency and feasibility for model selection of proposed variational methods using real networks.

We begin with a review of ERGMs, methods commonly used for fitting them and the adjusted pseudolikelihood in Section \ref{sec:ERGM}. Section \ref{sec:VB} describes the Bayesian variational approach. NCVMP and stochastic variational inference (SVI) are developed in Sections \ref{sec:NCVMP} and \ref{sec_SVI} respectively. Section \ref{sec:Model_sel} describes how model selection for ERGMs can be performed using variational methods and Section \ref{sec:Appl} presents the experimental results. Section \ref{sec:Conclu} concludes with a discussion.

\section{Exponential random graph models} \label{sec:ERGM}
Let $\N = \{1, \dots, n\}$ and $Y$ denote the $n \times n$ adjacency matrix of a network with $n$ nodes, where $Y_{ij}$ is 1 if there is a link from node $i$ to node $j$ and 0 otherwise. We assume that there are no self-links and hence $Y_{ii} = 0$ for $i \in \N$. If the network is undirected, then $Y$ is symmetric. Let $\mathcal{Y}$ denote the set of all possible networks on $n$ nodes and $y \in \mathcal{Y}$ be an observation of $Y$. In an ERGM, the likelihood of $y$ is 
\begin{equation*}
p(y|\theta) = \frac{\exp\{\theta^T s(y)\}}{z(\theta)}, \quad \text{where} \quad z(\theta) = \sum_{y \in \mathcal{Y}} \exp\{\theta^T s(y)\}
\end{equation*}
is the {\em normalizing constant}, $\theta \in \mathbb{R}^p$ is the vector of parameters and $s(y) \in \mathbb{R}^p$ is the vector of {\em sufficient statistics} for $y$, such as number of edges, number of triangles or nodal attributes. The normalizing constant, and hence the likelihood, cannot be evaluated except for trivially small graphs as it involves a sum over all networks in $\mathcal{Y}$ and the size of $\mathcal{Y}$ increases exponentially with $n$. The total possible number of undirected networks on $n$ nodes is $2^{n \choose 2}$. Let $\D$ denote the set of all dyads, where $\D = \{(i,j)|i, j \in \N, i<j\}$ for undirected networks, and $\D=\{(i,j)|i, j \in \N, i \neq j \}$ for directed networks.

In Bayesian inference, prior information about $\theta$ is captured by placing a prior on $\theta$. We consider $\theta \sim N(\mu_0, \Sigma_0)$ and a vague prior can be specified by setting $\Sigma_0 = \sigma_0^2 I_p$ with a large $\sigma_0^2$. The posterior density is $p(\theta|y) = p(y|\theta) p(\theta)/p(y)$, where $p(y) = \int p(y|\theta) p(\theta) d\theta$ is the marginal likelihood. Finding the posterior is a {\it doubly intractable} problem as normalizing constants in the likelihood and posterior are both intractable.

\subsection{Markov chain Monte Carlo maximum likelihood estimation} \label{sec_MCMLE}
The conventional method for estimating $\theta$ is MCMC MLE. As the likelihood is intractable, it is not maximized directly. Instead, the log ratio of the likelihoods at $\theta$ and some initial estimate $\theta_0$ is maximized. Note that
\begin{equation} \label{ztheta}
\frac{z(\theta)}{z(\theta_0)} 
= \sum_{y \in \mathcal{Y}} \exp\{ (\theta - \theta_0)^T s(y) \}  p(y|\theta_0) 
= E_{y|\theta_0} [ \exp\{s(y)^T  (\theta - \theta_0) \} ].
\end{equation}
Suppose  $\{y_1, \dots, y_K\}$ are networks simulated  from $p(y|\theta_0)$ via MCMC, then
\begin{equation*}
\text{LR}_{\theta_0}(\theta) 
= \log \frac{p(y|\theta)}{p(y|\theta_0)} 
\approx s(y)^T  (\theta - \theta_0) - \log \bigg[ \frac{1}{K} \sum_{k=1}^K \exp\{s(y_k)^T  (\theta - \theta_0) \} \bigg],
\end{equation*}
which can be maximized using Newton-Raphson or stochastic approximation. The value $\hat{\theta}_\ML$ at which $\text{LR}_{\theta_0}(\theta)$ is maximized serves as a maximum likelihood estimate of $\theta$.

To simulate from $p(y|\theta_0)$, \cite{Snijders2002} propose a Metropolis-Hastings algorithm, which begins with a network, $y^{(0)}$ (e.g. observed network), and randomly selects a dyad for toggling at each iteration. In the ``tie no tie" sampler, the dyad is randomly selected with equal probability from the set of dyads with ties or the set without ties. This reduces the probability of selecting a dyad without a tie, and improves mixing in the MCMC chain, as the proposal to toggle it has a high probability of rejection due to sparsity of most networks. Let $y_{-ij}$ denote the value of all dyads in $\D$ except $(i,j)$. Given $y^{(t-1)}$, the acceptance probability for toggling the value $y_{ij}^{(t-1)}$ of a candidate $(i,j)$ at iteration $t$ is 
\begin{equation*}
\min\left(1, \; \frac{p(y_{ij} \neq y_{ij}^{(t-1)}|y_{-ij} = y_{-ij}^{(t-1)}, \theta_0)}{p(y_{ij} = y_{ij}^{(t-1)}|y_{-ij} = y_{-ij}^{(t-1)}, \theta_0)} \right).
\end{equation*} 
The Metropolis-Hastings algorithm produces a sequence of networks $\{y^{(0)}, \dots, y^{(T)}\}$. The first portion is highly dependent on the initial network and is usually discarded as burn-in. These are the {\it auxiliary iterations} required before a simulation can be obtained from $p(y|\theta_0)$. A high thinning factor is often imposed to reduce correlation among retained samples. Simulating from $p(y|\theta_0)$ is thus computationally intensive.

\subsection{Pseudolikelihood}
\cite{Strauss1990} approximate the intractable likelihood of an ERGM with a {\em pseudolikelihood}, which assumes that the dyads are conditionally independent given the rest of the network. The pseudolikelihood is
\begin{equation*}
\begin{aligned}
f_\PL(y|\theta) &= \prod_{(i,j) \in \D} p(y_{ij}|y_{-ij}, \theta) = \prod_{(i,j) \in \D} \frac{ p(y_{ij} = 1|y_{-ij}, \theta)^{y_{ij}}}{p(y_{ij} = 0|y _{-ij}, \theta)^{y_{ij} - 1} }.
\end{aligned}
\end{equation*}
Now
\begin{equation} \label{logit}
\text{logit} \{ p(y_{ij} = 1|y_{-ij}, \theta) \} 
= \log \frac{p(y_{ij} = 1|y_{-ij}, \theta)}{p(y_{ij} = 0|y_{-ij}, \theta)} \\
= \theta^T \delta_s(y)_{ij},
\end{equation}
where $\delta_s(y)_{ij} = s(y_{ij}^+) - s(y_{ij}^-)$ is the vector of {\it change statistics} associated with $(i,j)$ and it represents the change in sufficient statistics when $y_{ij}$ is toggled from 0 ($y_{ij}^-$) to 1 ($y_{ij}^+$), with the rest of the network unchanged. From \eqref{logit},
\begin{equation*}
\begin{aligned}
\log f_\PL (y|\theta) &= \sum_{(i,j) \in \D} [ y_{ij} \theta^T \delta_s(y)_{ij} - \log \{ 1 + \exp(\theta^T \delta_s(y)_{ij}) \} ].
\end{aligned}
\end{equation*}
Maximization of the pseudolikelihood in place of the true likelihood can be performed efficiently using logistic regression where $\{y_{ij}\}$ are taken as responses and $\{\delta_s(y)_{ij}\}$ as predictors. However, this approach relies on the strong and often unrealistic assumption of conditionally independent dyads. Properties of the pseudolikelihood are not well understood \citep{vanDujin2009} and its use may also lead to biased estimates.

\subsection{Adjusted pseudolikelihood} \label{sec_APL}
\cite{Bouranis2018} propose an {\it adjusted pseudolikelihood} for correcting the mode, curvature and magnitude of the pseudolikelihood, which is given by 
\begin{equation*}
\begin{aligned}
\tilde{f}(y|\theta) &= M f_\PL(y|g(\theta)).
\end{aligned}
\end{equation*}
Th constant $M > 0$ adjusts the magnitude and $g: \mathbbm{R}^p \rightarrow \mathbbm{R}^p$ is an invertible affine transformation that adjusts the mode and curvature of the pseudolikelihood to match the true likelihood. It is defined as 
\begin{equation} \label{g}
g(\theta) = \hat{\theta}_{\PL} + W(\theta - \hat{\theta}_{\ML}),
\end{equation}
where $\hat{\theta}_{\ML} = \argmax_\theta p(y|\theta)$ is the maximum likelihood estimate, $\hat{\theta}_{\PL} = \argmax_\theta f_{\PL}(y|\theta)$ is the maximum pseudolikelihood estimate and $W$ is a $p \times p$ upper triangular matrix. As 
\begin{equation*}
\argmax_\theta \tilde{f}(y|\theta) = \argmax_\theta f_\PL(y|g(\theta)) = g^{-1} ( \hat{\theta}_{\PL}) = \hat{\theta}_{\ML},
\end{equation*}
the adjusted pseudolikelihood has the same mode as the true likelihood. 

The matrix $W$ is selected so that $\log \tilde{f}(y|\theta)$ has the same curvature as the true log likelihood at the mode. The Hessian of $\log \tilde{f}(y|\theta)$ is $W^T \nabla_\theta^2 \log f_\PL(y|g(\theta)) W$ and the Hessian of $\log p(y|\theta)$ is $-\cov_{y|\theta}[s(y)]$, where $\cov_{y|\theta}[s(y)]$ is the covariance matrix of $s(y)$ with respect to $p(y|\theta)$. Details are given in the supplementary material. As $g( \hat{\theta}_{\ML} ) = \hat{\theta}_\PL$, 
\begin{equation*}
W^T \nabla_\theta^2 \log f_\PL(y|\hat{\theta}_{\PL}) W = -\cov_{y|\hat{\theta}_{\ML}}[s(y)].
\end{equation*}
Let $R_1^TR_1$ and $R_2^TR_2$ be unique Cholesky decompositions (with positive diagonal entries) of $- \nabla_\theta^2 \log f_\PL(y|\hat{\theta}_{\PL})$ and $\cov_{y|\hat{\theta}_{\ML}} [s(y)]$ respectively, where $R_1$ and $R_2$ are $p \times p$ upper triangular matrices. Then $W^T R_1^TR_1 W = R_2^TR_2$. By uniqueness, $W = R_1^{-1} R_2$. We can estimate $\cov_{y|\hat{\theta}_{\ML}}[s(y)]$ using Monte Carlo by simulating from $p(y|\hat{\theta}_{\ML})$.

Finally, $f_\PL(y|g(\theta))$ is scaled by $M$ to have the same magnitude as the true likelihood at the mode. This implies that
\begin{equation*}
M = \frac{p(y|\hat{\theta}_{\ML})}{f_\PL(y|g(\hat{\theta}_{\ML})) }= \frac{\exp(\hat{\theta}_{\ML}^T s(y))/z(\hat{\theta}_{\ML})}{f_\PL(y|\hat{\theta}_{\PL}) }.
\end{equation*}
\cite{Bouranis2018} propose an importance sampling procedure to estimate $z(\hat{\theta}_{\ML})$. Introducing a sequence of temperatures $0= t_0 < t_1 < \dots < t_J=1$,
\begin{equation*}
z(\hat{\theta}_{\ML}) = z(0)\prod_{j=1}^{J} \frac{z(t_{j}\hat{\theta}_{\ML})}{z(t_{j-1}\hat{\theta}_{\ML})},
\end{equation*}
where $z(0) = 2^{n\choose 2}$ for undirected networks. Each ratio is then estimated using importance sampling. From \eqref{ztheta},
\begin{equation*}
\begin{aligned}
\frac{z(t_{j}\hat{\theta}_{\ML})}{z(t_{j-1}\hat{\theta}_{\ML})} \approx \frac{1}{K}\sum_{k=1}^K \exp\{ (t_{j} - t_{j-1}) \hat{\theta}_{\ML}^T s(y_k^{(j-1)}) \},
\end{aligned}
\end{equation*}
where $\{y_1^{(j-1)}, \dots, y_K^{(j-1)}\}$ are samples from $p(y|t_{j-1}\hat{\theta}_{\ML})$. Similar estimators can also be obtained using annealed importance sampling \citep{Neal2001}.

The above procedure hinges on $z(0)$ being known and slowly shifts this value towards $z(\hat{\theta}_\ML)$. While the procedure works well for small networks, it is hard to implement for large networks as sampling from $p(y|t_{j-1}\hat{\theta}_{\ML})$ for a small $t_{j-1}$ is difficult for large $n$. For instance, when $j=1$, we need to draw uniformly from the set of all possible networks. The inclusion probability of each edge is 0.5 and the average network size is $0.5n(n-1)$, which is large for large $n$. If the Metropolis-Hastings algorithm in Section \ref{sec_MCMLE}, which initializes with the observed sparse network, is used for simulation, it will take a large number of samples for a simulated network to reach the average size. Biased estimates may result if the burn-in is not long enough. We propose a modification which can be applied if the first sufficient statistic is {\em number of edges}, $\sum_{(i,j) \in \D} y_{ij}$. Let $\hat{\theta}_{\ML, 1}$ be the first element of $\hat{\theta}_{\ML}$, and $\hat{\theta}_{\ML, -1}$ denote $\hat{\theta}_{\ML}$  with the first element removed. The idea is to fix the first element at $\hat{\theta}_{\ML, 1}$ and let the remaining elements approach $\hat{\theta}_{\ML, -1}$ starting from zero. As most observed networks are sparse, $\hat{\theta}_{\ML,1}$ is often small and helps to control the size of simulated networks. We have
\begin{equation*}
z(\hat{\theta}_{\ML}) = z([\hat{\theta}_{\ML, 1}, 0])\prod_{j=1}^{J} \frac{z([\hat{\theta}_{\ML, 1}, t_{j}\hat{\theta}_{\ML, -1}])}{z([\hat{\theta}_{\ML, 1}, t_{j-1}\hat{\theta}_{\ML, -1}])},
\end{equation*}
where $z([\hat{\theta}_{\ML, 1}, 0]) $ is the normalizing constant of a network where the only sufficient statistic is number of edges. The likelihood of this dyad independent network is $p(y|\hat{\theta}_{\ML, 1}) =\prod_{(i,j) \in \D} \exp(\hat{\theta}_{\ML, 1}y_{ij})/\{ 1 + \exp(\hat{\theta}_{\ML, 1})\}$. For undirected networks,
\begin{equation*}
\log z([\hat{\theta}_{\ML, 1}, 0]) = 0.5 n(n-1) \log \{1+ \exp(\hat{\theta}_{\ML, 1})\}.
\end{equation*}
We can again estimate each ratio using importance sampling by
\begin{equation*}
\begin{aligned}
\frac{z([\hat{\theta}_{\ML, 1}, t_{j}\hat{\theta}_{\ML, -1}])}{z([\hat{\theta}_{\ML, 1}, t_{j-1}\hat{\theta}_{\ML, -1}])} \approx \frac{1}{K}\sum_{k=1}^K \exp\{ (t_{j} - t_{j-1}) \hat{\theta}_{\ML,-1}^T s(y_k^{(j-1)})_{-1} \},
\end{aligned}
\end{equation*}
where $\{y_1^{(j-1)}, \dots, y_K^{(j-1)}\}$ are samples from $p(y|[\hat{\theta}_{\ML, 1}, t_{j-1}\hat{\theta}_{\ML, -1}])$ and $s(y)_{-1}$ denotes the vector of sufficient statistics excluding the first element. 

\cite{Bouranis2018} used the adjusted pseudolikelihood in place of the true likelihood in MCMC algorithms and showed that the Bayes factor for performing model selection can be estimated accurately with reduced computation. Next, we develop variational inference methods for the ERGM, one of which uses this adjusted pseudolikelihood.

\section{Bayesian variational inference}\label{sec:VB}
In Bayesian variational inference, the true posterior of $\theta$ is approximated by a more tractable density, $q_\lambda(\theta)$, with parameters $\lambda$. It is commonly assumed that $q_\lambda(\theta)$ belongs to a parametric family or is of a factorized form, say $q_\lambda(\theta)= \prod_{i=1}^p q_i(\theta_i)$. The KL divergence between the variational density and true posterior,
\begin{equation*}
\text{KL}[q_\lambda(\theta)||p(\theta|y)] = \int q_\lambda(\theta) \log \frac{q_\lambda(\theta)}{p(\theta|y)} d\theta,
\end{equation*}
is then minimized subject to these restrictions. We consider a Gaussian approximation $N(\mu, \Sigma)$ of the posterior density, where $\lambda$ denotes the parameters $\{\mu, \Sigma\}$. This assumption allows posterior correlation among elements of $\theta$ to be captured and $q_\lambda(\theta)$ is likely to approximate the true posterior well so long as the Gaussian assumption is not strongly violated. Posterior estimation is thus reduced to an optimization problem of finding $\lambda$ that minimizes the KL divergence. As $\text{KL}[q_\lambda(\theta)||p(\theta|y)] \geq 0$,
\begin{equation} \label{LB}
\log p(y) = \underbrace{\int q_\lambda(\theta) \log \frac{p(\theta,y)}{q_\lambda(\theta)} d\theta}_{\L} \;+ \;\underbrace{\int q_\lambda(\theta)\log \frac{q_\lambda(\theta)} {p(\theta|y)}d\theta}_{\text{KL divergence}} \;\geq \;\L.
\end{equation}
The log marginal likelihood is bounded below by $\L$, the {\em evidence lower bound}. Minimizing the KL divergence is thus equivalent to maximizing $\L$ with respect to $\lambda$. 

For ERGMs, $\L$ is intractable due to the likelihood. We propose two approaches to overcome this problem. The first plugs in the adjusted pseudolikelihood for the true likelihood and intractable expectations are approximated (deterministically) using Gauss-Hermite quadrature \citep{Liu1994}.  As the adjusted pseudolikelihood is nonconjugate with respect to the prior of $\theta$, we optimize $\L$ using nonconjugate variational message passing \citep[NCVMP,][]{Knowles2011}. In the second approach, we consider stochastic variational inference \citep[SVI,][]{Titsias2014}, which does not require expectations to be evaluated analytically. A reparametrization trick is applied and $\L$ is optimized using stochastic gradient ascent. The gradients are estimated using Monte Carlo or self-normalized importance sampling.

\section{Nonconjugate variational message passing} \label{sec:NCVMP}
If $q_\lambda(\theta)= \prod_{i=1}^p q_i(\theta_i)$ and each $q_i$ belongs a exponential family, \cite{Winn2005} showed that, for conjugate-exponential models, optimizing each $q_i$ involves only a local computation at the node $\theta_i$. A term from the parent nodes and one term from each child node of $\theta_i$ are summed, and these terms can be interpreted as ``messages" passed from the neighboring nodes, hence ``variational message passing". \cite{Knowles2011} consider an extension to nonconjugate models by approximating intractable expectations using bounds or quadrature. We assume that $q_\lambda(\theta)$ belongs to an exponential family (Gaussian) but do not consider a factorized form, which may result in underestimation of the posterior variance. Thus the ``messages" passed to $\theta$ consist only of one from the parent nodes $\{\mu_0, \Sigma_0\}$ and one from the child node $y$, as illustrated in Figure \ref{factor_graph}.
 \begin{figure}[htb!]
\centering
\includegraphics[width=0.4\textwidth]{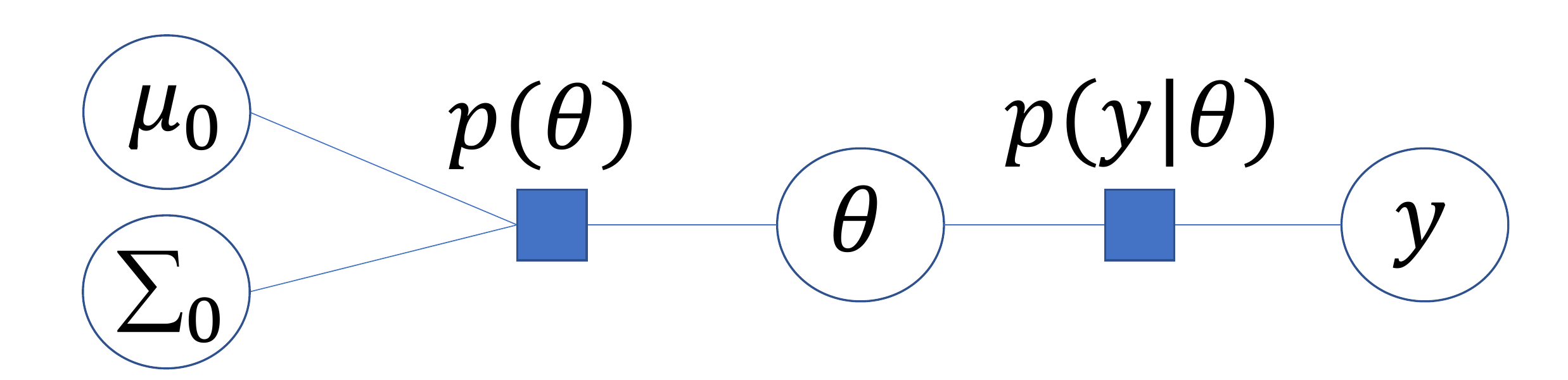}
\caption{Factor graph for ERGM. Filled rectangles denote factors. \label{factor_graph}}
\end{figure}

Suppose $q_\lambda(\theta) = \exp\{ \lambda^T t(\theta) - h(\lambda)\}$, where $\lambda$ is the vector of natural parameters and $t(\cdot)$ are the sufficient statistics. Let $E_{q_\lambda}$ denote expectation with respect to $q_\lambda$. From \eqref{LB},
\begin{equation} \label{euclidean_gradient}
\begin{gathered}
\L = E_{q_\lambda} \{ \log p(\theta,y) \} - \nabla_\lambda h(\lambda)^T \lambda + h(\lambda), \\
\nabla_\lambda \L= \nabla_\lambda E_{q_\lambda} \{ \log p(\theta,y) \} - \V(\lambda) \lambda,
\end{gathered}
\end{equation}
where $E_{q_\lambda}\{ t(\theta) \} = \nabla_\lambda h(\lambda)$ and $\V(\lambda) = \nabla^2_\lambda h(\lambda)$ is the covariance matrix of $t(\theta)$ with respect to $q_\lambda$. To maximize $\L$, we set $\nabla_\lambda \mathcal{L}$ to zero, which leads to the update,
\begin{equation*}
\lambda \leftarrow \V(\lambda)^{-1} \nabla_\lambda E_{q_\lambda} \{ \log p(y, \theta)\}.
\end{equation*}
As $\log p(y, \theta) =\log p(y|\theta) + \log p(\theta)$, the update is a sum of messages from the neighboring factors. If $q_\lambda(\theta)$ is $N(\mu, \Sigma)$, then the update for $\lambda$ simplifies to 
\begin{equation} \label{ncvmpupdate}
\begin{aligned}
\Sigma &\leftarrow -\frac{1}{2}\left( \vec^{-1} \left[ \nabla_{\vec(\Sigma)} E_{q_\lambda} \{ \log p(y, \theta)\} \right]  \right)^{-1}, \quad \mu \leftarrow \mu + \Sigma \nabla_\mu E_{q_\lambda} \{ \log p(y, \theta)\}.
\end{aligned}
\end{equation}
Details can be found in \cite{Tan2013} and \cite{Wand2014}. Note that $a = \vec(A)$ is a vector obtained by stacking the columns of matrix $A$ under each other from left to right and $\vec^{-1}(a)$ recovers matrix $A$ from $a$. As a fixed point iteration algorithm, NCVMP is not guaranteed to converge and the lower bound may not necessarily increase after each update. However, if the algorithm converges, then it will be to a local maximum. Convergence issues can be addressed by adjusting the initialization or using damping. More details are given in the supplementary material. 

Let $\tilde{p}(y, \theta) = \tilde{f}(y|\theta) p(\theta)$ denote the joint density obtained by plugging in the adjusted pseudolikelihood for $p(y|\theta)$.
Let $b(x) = \log\{1+\exp(x)\}$, $\alpha_{ij} = \delta_s(y)_{ij}^T ( \hat{\theta}_{\PL}  - W\hat{\theta}_{\ML})$ and $\beta_{ij} = W^T \delta_s(y)_{ij}$ so that $\delta_s(y)_{ij}^Tg(\theta) = \alpha_{ij} + \beta_{ij} \theta$. Then 
\begin{equation} \label{logjoint}
\begin{aligned}
E_{q_\lambda} \{\log \tilde{p}(y, \theta) \} &= \log M + \sum_{(i,j) \in \D} [y_{ij} (\alpha_{ij} + \beta_{ij}^T \mu)  - E_{q_\lambda} \{b(\alpha_{ij} + \beta_{ij}^T \theta) \}] \\
& -  \tfrac{p}{2}\log (2\pi) - \tfrac{1}{2}\log|\Sigma_0| -  \tfrac{1}{2}(\mu - \mu_0)^T \Sigma_0^{-1}(\mu - \mu_0) -  \tfrac{1}{2}\tr(\Sigma_0^{-1} \Sigma).
\end{aligned}
\end{equation}
The approximate lower bound $\tilde{\L} = E_{q_\lambda} \{\log \tilde{p}(y, \theta) - \log q_\lambda(\theta)\}$ and other derivation details are given in the supplementary material. The term $E_{q_\lambda}\{ b(\alpha_{ij} +\beta_{ij}^T\theta) \}$ is approximated using Gauss-Hermite quadrature. Now $\theta \sim N(\mu, \Sigma)$ if and only if  $\alpha_{ij} + \beta_{ij}^T\theta \sim N(m_{ij}, v_{ij}^2)$, where $m_{ij} = \alpha_{ij} + \beta_{ij}^T\mu$ and $v_{ij}^2 = \beta_{ij}^T \Sigma \beta_{ij}$. Let $b^{r} (\cdot)$ denote the $r$th derivative of $b(\cdot)$ and define $B^{(r)}(m_{ij}, v_{ij}) = E_{q_\lambda}\{ b^{(r)}(\alpha_{ij} +\beta_{ij}^T\theta) \}$. Then
\begin{equation*}
\begin{aligned}
B^{(r)}(m_{ij}, v_{ij}) =  \int b^{(r)}(\alpha_{ij} +\beta_{ij}^T\theta) q_\lambda(\theta) d\theta
=  \int_{-\infty}^\infty b^{(r)}(x) \phi(x|m_{ij}, v_{ij}) dx 
&=  \int_{-\infty}^\infty g_{ij}^{(r)}(z)  dz,
\end{aligned}
\end{equation*}
where $\phi(x|m, v)$ denotes the density of $N(m, v^2)$ and $g_{ij}^{(r)}(z) =  b^{(r)}(v_{ij} z + m_{ij}) \phi(z|0,1)$. This conversion of $B^{(r)}(m_{ij}, v_{ij})$ from a multivariate to a univariate integral was proposed in \cite{Ormerod2012} and used in \cite{Tan2013}. Let $\{x_d\}_{d=1}^D$ be zeros of the $D$th order Hermite polynomial and $\{w_d\}_{d=1}^D$ be the corresponding weights. In Gauss-Hermite quadrature, $\int_{-\infty}^\infty f(x)\e^{-x^2} dx \approx \sum_{d=1}^D w_d f(x_d)$. Following \cite{Liu1994}, we first transform $z$ so that the integrand $g^{(r)}_{ij}(z)$ is sampled in a suitable range. Let $\hat{m}_{ij}^{(r)}$ be the mode of $g^{(r)}_{ij}(z)$, $(\vijr)^{-2}= - \partial^2 \log g_{ij}^{(r)}(z)/\partial^2 z |_{z=\hat{m}_{ij}}$ and $w_d^* = w_d \exp(x_d^2)$ be the modified weights. We have
\begin{multline*}
B^{(r)}(m_{ij}, v_{ij}) = \int_{-\infty}^\infty \frac{g_{ij}^{(r)}(z)}{\phi(z|\mijr, \vijr)} \phi(z|\mijr, \vijr)dz \\
= \sqrt{2}\vijr  \int_{-\infty}^\infty \e^{x^2} g_{ij}^{(r)}(\sqrt{2}\vijr x + \mijr) \e^{-x^2}dx
\approx \sqrt{2} \vijr  \sum_{d=1}^D w_d^* g_{ij}^{(r)} (\sqrt{2} \vijr x_d + \mijr ).
\end{multline*}

From \eqref{logjoint}, differentiating $E_{q_\lambda} \{\log \tilde{p}(y, \theta)\} $ with respect to $\mu$ and $\vec(\Sigma)$, 
\begin{equation*}
\begin{gathered}
\nabla_\mu  E_{q_\lambda} \{\log \tilde{p}(y, \theta)\} =
\sum_{(i,j) \in \D} \{ y_{ij} - B^{(1)}(m_{ij}, v_{ij}) \} \beta_{ij} - \Sigma_0^{-1}(\mu - \mu_0), \\
\nabla_{ \vec(\Sigma)} E_{q_\lambda} \{\log \tilde{p}(y, \theta)\} = -\frac{1}{2}\vec\bigg(\Sigma_0^{-1} + \sum_{(i,j) \in \D} B^{(2)}(m_{ij}, v_{ij}) \beta_{ij} \beta_{ij}^T \bigg).
\end{gathered}
\end{equation*}
Details are given in the supplementary mterials. Substituting these gradients in \eqref{ncvmpupdate}, we obtain the NCVMP algorithm.

\floatname{algorithm}{NCVMP Algorithm}
\begin{algorithm}
\caption{}
\begin{enumerate}
\itemsep-0.1em 
\item Find the nodes $\{x_d\}_{d=1}^D$ of $D$th order Hermite polynomial and weights $\{w_d^*\}_{d=1}^D$.
\item Find adjusted pseudolikelihood parameters: $\hat{\theta}_\PL$, $\hat{\theta}_\ML$, $M$, $W$ and compute $\{\alpha_{ij}\}$, $\{\beta_{ij}\}$. 
\item Initialize $\mu = \hat{\theta}_{ML}$, $\Sigma= 0.01I_p$ and compute $\{m_{ij}\}$,  $\{v_{ij}\}$, $\tilde{\L}^{\text{old}}$. Set $\epsilon=1$.
\item While $\epsilon > \text{tolerance}$,
\begin{enumerate}[i.]
\vspace*{-2mm}
\item Update $\Sigma \leftarrow \big( \Sigma_0^{-1} + \sum_{(i,j) \in \D} B^{(2)}(m_{ij}, v_{ij}) \beta_{ij} \beta_{ij}^T \big)^{-1}$.
\item Update $\mu \leftarrow \mu + \Sigma \big[ \sum_{(i,j) \in \D} \{ y_{ij} - B^{(1)}(m_{ij}, v_{ij}) \} \beta_{ij} - \Sigma_0^{-1}(\mu - \mu_0)  \big]$.
\item Update $m_{ij} \leftarrow \alpha_{ij} + \beta_{ij}^T\mu$ and $v_{ij}^2 \leftarrow \beta_{ij}^T \Sigma \beta_{ij}$ for all $(i,j) \in \D$.
\item Compute new lower bound $\tilde{\L}^{\text{new}}$ and $\epsilon = (\tilde{\L}^{\text{new}} - \tilde{\L}^{\text{old}})/ |\tilde{\L}^{\text{old}}|$. $\tilde{\L}^{\text{old}} \leftarrow \tilde{\L}^{\text{new}}$.
\end{enumerate}
\end{enumerate}
\vspace{-3mm}
\end{algorithm}

The nodes and weights in step 1 can be obtained in Julia using {\tt gausshermite} from the package {\tt FastGaussQuadrature}, and we set $D=20$. The NCVMP algorithm is not guaranteed to converge to a local maximum as it is a fixed-point iteration method and also due to the approximation of expectations using Gauss-Hermite quadrature. However, we can compute $\tilde{\L}$ at each iteration to check that the algorithm is moving towards a local maximum. If $\tilde{\L}$ does not increase, we can use damping. The algorithm is terminated when the relative increase in $\tilde{\L}$ is negligible, with tolerance set as $10^{-5}$. NCVMP can be sensitive to the initialization. Here we initialize $\mu$ as $\hat{\theta}_{\ML}$, an informative starting point.

\section{Stochastic variational inference} \label{sec_SVI}
Instead of using Gauss-Hermite quadrature to approximate intractable expectations, we can optimize $\L$ with respect to $\lambda$ using stochastic gradient ascent \citep{Robbins1951}. Let $CC^T$ be the unique Cholesky decomposition of $\Sigma$, where $C$ is a $p \times p$ lower triangular matrix with positive diagonal entries, and $\vech(A)$ denote the vector obtained by vectorizing the lower triangular part of a square matrix $A$. At each iteration $t$, 
\begin{equation}\label{SG_update}
\mu^{(t+1)} = \mu^{(t)} + \rho_t \hat{\nabla}_\mu \L^{(t)}, \quad
\vech(C^{(t+1)}) = \vech(C^{(t)}) + \rho_t \hat{\nabla}_{\vech(C)} \L^{(t)},
\end{equation}
where $\hat{\nabla}_\mu \L^{(t)}$ and $\hat{\nabla}_{\vech(C)} \L^{(t)}$ are unbiased estimates of $\nabla_\mu \L$ and $\nabla_{\vech(C)} \L$ respectively and $\rho_t$ denotes the step-size. Convergence is ensured if some regularity conditions are fulfilled and the step size satisfies $\rho_t\rightarrow 0$, $\sum_t \rho_t= \infty$, $\sum_t \rho_t^2 < \infty$ \citep{Spall2003}. 

As $\L$ is an expectation with respect to $q_\lambda$, unbiased gradients can be obtained by simulating $\theta$ from $q_\lambda(\theta)$. We use the {\em reparametrization trick} and apply the transformation $\theta = Cs + \mu$, where $s \sim N(0, I_p)$ and has density $\phi(s)$. Then
\begin{equation} \label{LB2}
\mathcal{L} = E_{\phi} \{ \log p(\theta,y)\} - E_{\phi} \{  \log q_\lambda(\theta) \}, 
\end{equation}
where $\theta = Cs + \mu$ and $E_\phi$ denotes expectation with respect to $\phi(s)$. This moves the variational parameters inside the expectation so that the stochastic gradients are direct functions of $\{\mu, C\}$. Unbiased gradient estimates can be obtained by simulating $s \sim \phi(s)$. The reparametrization trick does not always reduce the variance of the stochastic gradients \citep{Gal2016}. However, \cite{Xu2019} show that the variance of stochastic gradients obtained using this trick are smaller than that obtained using the score function \citep{Williams1992} under a mean-field variational Bayes Gaussian approximation, if the log joint density is a quadratic function centered at the variational mean. 

Although $E_{\phi}\{ \log q_\lambda(\theta)\}$ can be evaluated analytically, estimating both terms in \eqref{LB2} using the same samples $s \sim \phi(s)$ allow the stochasticity from $s$ in the two terms to cancel out so that there is smaller variation in the gradients at convergence \citep{Roeder2017,Tan2018, Tan2018b}. As $\log q_\lambda(\theta)$ depends on $\{ \mu, C\}$ directly as well as through $\theta$, we apply the chain rule to obtain
\begin{align}
\nabla_{\mu} \L &= E_{\phi} \{\nabla_{\theta} \log p(\theta, y) - \nabla_{\theta} \log q_\lambda(\theta) - \nabla_{\mu} \log q_\lambda(\theta)\},  \label{g1} \\
\nabla_{\vech(C)} \L &= E_{\phi} [\vech \{ \nabla_{\theta}\log p(\theta, y)s^T - \nabla_{\theta} \log q_\lambda(\theta) s^T\} - \nabla_{\vech(C)} \log q_\lambda(\theta) ],  \label{g2} 
\end{align}
where $- \nabla_{\theta} \log q_\lambda(\theta) = \nabla_{\mu} \log q_\lambda(\theta) = C^{-T} s$ and $\nabla_{\vech(C)} \log q_\lambda(\theta) =  \vech(C^{-T} (ss^T - I)]$. Derivations are given in the supplementary material. The last term in \eqref{g1} and \eqref{g2} together represent the score of $q_\lambda$, whose expectation is zero. Hence we can omit these terms to construct unbiased gradient estimates,
\begin{equation*}
\begin{aligned}
\hnabla_{\mu} \L = \nabla_{\theta}\log p(\theta, y)+ C^{-T} s, \quad
\hnabla_{\vech(C)} \L = \vech \{ \hnabla_{\mu} \L s^T \}.
\end{aligned}
\end{equation*}
Omitting the last term in \eqref{g1} and \eqref{g2} yield better results as gradients constructed in this way are approximately zero at convergence \citep{Tan2018b}.

The update for $\vech(C)$ in \eqref{SG_update} does not ensure diagonal entries of $C$ remain positive. Hence we introduce lower triangular matrix $C'$, where $C'_{ii} = \log(C_{ii})$ and $ C'_{ij} = C_{ij}$ if $i \neq j$, and update $\vech(C')$ instead. Let $D_C = \diag\{ \vech( \tilde{C})\}$ where $\tilde{C}$ is a $p \times p$ matrix with the diagonal of $C$ and ones everywhere else. Then $\nabla_{\vech(C')} \L = D_C \nabla_{\vech(C)} \L$. 

At each iteration $t$, given $\theta^{(t)} = C^{(t)}s^{(t)} +\mu^{(t)}$ for $s^{(t)} \sim N(0, I_p)$, we need to compute
\begin{equation*}
\nabla_\theta \log p(\theta^{(t)} , y)  = s(y) - E_{y|\theta^{(t)} } [s(y)]- \Sigma_0^{-1} (\theta ^{(t)} - \mu_0).
\end{equation*}
Estimating $E_{y|\theta^{(t)}} [s(y)]$ is challenging because sampling from the likelihood is computationally intensive. We discuss two approaches below. The first retains unbiasedness of the gradients, while the second results in biased but consistent gradients.

\subsection{Monte Carlo sampling}
Given $\theta^{(t)}$ at iteration $t$, let $\mathcal{S}_K(\theta^{(t)}) = \{s(y_1^{(t)}), \dots, s(y_K^{(t)})\}$ be the set of sufficient statistics of $K$ networks, $\{y_1^{(t)}, \dots, y_K^{(t)}\}$, simulated from $p(y|\theta^{(t)})$. We can compute an unbiased Monte Carlo estimate, $E_{y|\theta^{(t)}} \{s(y)\} \approx \frac{1}{K} \sum_{k=1}^K s(y_k^{(t)})$. In stochastic approximation, we only require unbiased gradient estimates for convergence and $K$ need not be large. We investigate the performance of this approach for $K$ as small as one in our experiments.

\subsection{Self-normalized importance sampling}
Suppose $\mathcal{S}_K(\theta^u) = \{s(y_1^u), \dots, s(y_K^u)\}$ is the set of sufficient statistics of $\{y_1^u, \dots, y_K^u\}$ simulated from $p(y|\theta^u)$ for some $\theta^u \in \mathbb{R}^p$. From \eqref{ztheta}, 
\begin{multline*}
E_{y|\theta^{(t)}} \{s(y)\} 
= \frac{z(\theta^u)}{z(\theta^{(t)})} \sum_{y \in \mathcal{Y}}  s(y) \exp \{ s(y)^T (\theta^{(t)}-\theta^u) \} p(y|\theta^u)  \\
= \frac{E_{y|\theta^u} [s(y) \exp \{ s(y)^T (\theta^{(t)}-\theta^u) \} ]}{E_{y|\theta^u} [ \exp\{ s(y)^T (\theta^{(t)} - \theta^u) \} ]  }  
\approx \sum_{k=1}^K \tilde{w}_k^u s(y_k^u),
\end{multline*} 
where $w_k^u =\exp \{ s(y_k^u)^T (\theta^{(t)}-\theta^u) \} $ and $\tilde{w}_k^u = w_k^u/ (\sum_{k'=1}^K w_{k'}^u )$ is the {\em normalized weight}. The SNIS estimate is consistent by the strong law of large numbers but induces a small bias of $\mathcal{O}(1/K)$ \citep{Liu2004}, and is asymptotically unbiased. Convergence results in stochastic gradient descent usually require unbiased gradients, but biased gradients have been used in recent works for efficiency \citep[e.g.][]{Chen2018, Le2019}. \cite{Tadic2017} prove that iterates of stochastic gradient search using biased gradients converge to a neighborhood of the set of minima, conditional on the asymptotic bias of the gradient estimator. \cite{Chen2019} show that consistent but biased gradient estimators exhibit similar convergence behaviors as unbiased ones. These studies lend support that stochastic gradient ascent with biased but consistent gradients will converge to a vicinity of the optima.

SNIS alleviates the burden of sampling at every iteration and improves the convergence rate of the stochastic approximation algorithm tremendously. A good initial choice of $\theta^u$ is $\hat{\theta}_\ML$, but it is unlikely that SNIS based on the proposal of $p(y|\hat{\theta}_\ML)$ will work well for any $\theta^{(t)}$. For instance, SNIS may be poor if $\theta^{(t)}$ and $\hat{\theta}_\ML$ are far apart. One way of assessing how different the proposal is from the target distribution and the efficiency of the SNIS estimate is to use an approximation of the effective sample size \citep{Kong1994, Martino2017}, $\ESS = 1/{\sum_{k=1}^K (\tilde{w}_k^u)^2}$. As an example, consider $10^4$ networks  simulated from $p(y|\hat{\theta}_\ML)$ for the karate network (Model 2) in Section \ref{sec:karate}. Suppose SVI is performed with gradients estimated using SNIS with $p(y|\hat{\theta}_\ML)$ as proposal. From Figure \eqref{karate_IS}, the ESS decreases rapidly as the distance between $\theta^{(t)}$ and $\hat{\theta}_\ML$ increases. Moreover, the ESS is lower than 10000/3 for 4158 out of 7000 iterations suggesting very poor efficiency.
\begin{figure}[tb!]
\centering
\includegraphics[width=0.8\textwidth]{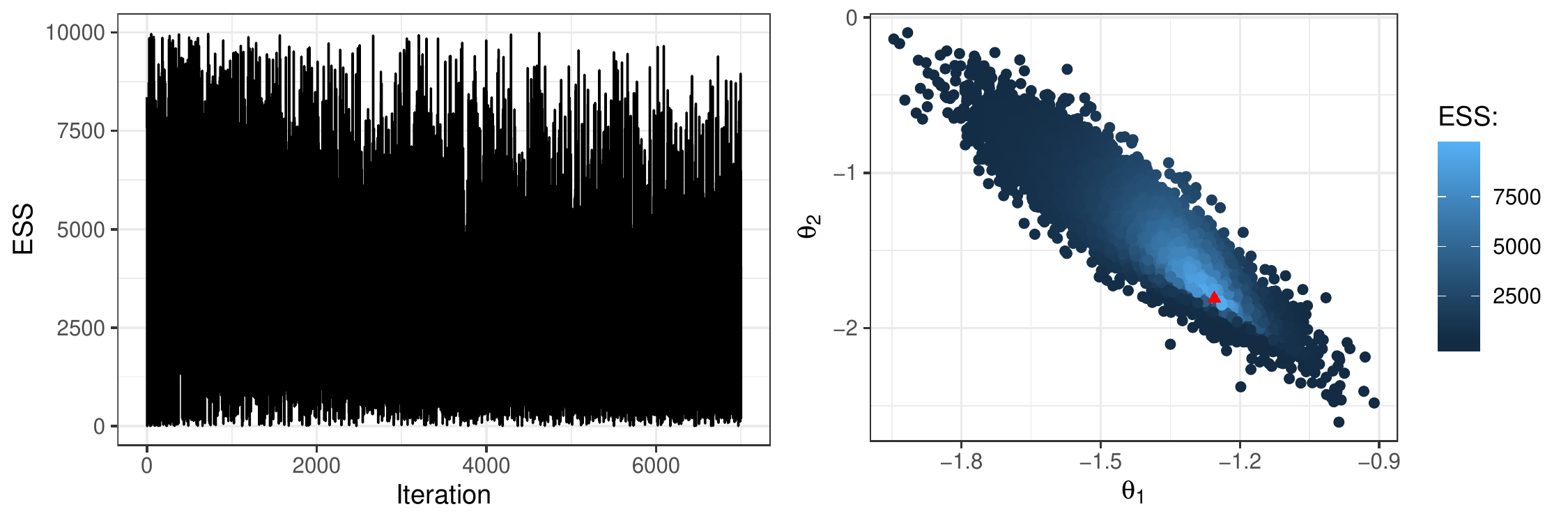}
\caption{\small Karate network. SVI is performed using SNIS based on $10^4$ networks simulated from $p(y|\hat{\theta}_\ML)$. Left: ESS at each iteration. Right: $\theta^{(t)}$ colored according to ESS at iteration $t$. Red triangle denotes $\hat{\theta}_\ML$.} \label{karate_IS}
\end{figure}

To minimize the cost of simulating from the likelihood and avoid poor SNIS estimates, we propose an {\em adaptive} sampling strategy, which maintains a collection $\mathbb{S}$ of particles and associated sufficient statistics at any iteration. Given $\theta^{(t)}$, normalized importance weights are computed using the particle $\theta^u \in \mathbb{S}$ closest to $\theta$. We use the Mahalanobis distance to measure closeness and the current estimate of $\Sigma =CC^T$ as the covariance matrix, as it is scale-invariant and takes into account correlations in the parameter space. Thus
\begin{equation*}
d_M(\theta^{(t)}, \theta^u) = \sqrt{(\theta^{(t)} - \theta^u)^T {C^{(t)}}^{-T} {C^{(t)}}^{-1}(\theta^{(t)} - \theta^u)}.
\end{equation*}
If the $\ESS$ is lower than some threshold, say $K/3$, $K$ networks are simulated from $p(y|\theta^{(t)})$ and $E_{y|\theta^{(t)}} \{s(y)\}$ is estimated using Monte Carlo. The particle $\theta^{(t)}$ and $\mathcal{S}_K(\theta^{(t)})$ are then added to $\mathbb{S}$. Otherwise, the SNIS estimate is used. We begin with just one particle $\hat{\theta}_\ML$ and $\mathcal{S}_K(\hat{\theta}_\ML)$, and allow $\mathbb{S}$ to grow as the algorithm proceeds. This strategy is likely to be more effective for low-dimensional problems, as a large number of particles may be required to cover the region where $q_\lambda(\theta)$ is practically nonzero if $p$ is large.

Related MCMC approaches include the adaptive exchange algorithm \citep{Liang2016}, which uses samples from pre-specified particles chosen using fractional double Metropolis-Hastings and a max-min procedure. \cite{Atchade2013} develop an adaptive MCMC algorithm for approximating $z(\theta)$ through a linear combination of importance sampling estimates based on multiple particles. The particles are selected using stochastic approximation recursion at the beginning and remain fixed throughout the algorithm.

\subsection{Diagnosing convergence and adaptive stepsize}
We use the evidence lower bound $\L$ to diagnose convergence of the SVI Algorithm. Suppose we simulate $s^{(t)}$ from $N(0, I_p)$ and $K_0$ samples from $p(y|\hat{\theta}_\ML)$, with sufficient statistics $\mathcal{S}_{K_0} (\hat{\theta}_\ML) = \{s(y_1), \dots, s(y_{K_0})\}$. From $\eqref{ztheta}$, an estimate of $\L$ at the $t$th iteration is 
\begin{equation} \label{LBest}
\begin{aligned}
\hat{\L}_t & = {\theta^{(t)}}^T s(y) - \log z(\hat{\theta}_\ML) - \log \bigg[ \frac{1}{K_0} \sum_{k=1}^{K_0} \e^{s(y_k)^T (\theta^{(t)}-\hat{\theta}_\ML)} \bigg] \\
& \quad -\tfrac{1}{2} \log|\Sigma_0| - \tfrac{1}{2}(\theta^{(t)} - \mu_0)^T \Sigma_0^{-1}(\theta^{(t)} - \mu_0) + \log|C^{(t)}| + \tfrac{1}{2}{s^{(t)}}^T s^{(t)},
\end{aligned}
\end{equation}
We set $K_0 =1000$ in our experiments. An estimate of $\log z(\hat{\theta}_\ML)$ can be obtained using the {\tt ergm} function from the {\tt ergm} R package \citep{Hunter2008} or the importance sampling procedure described in Section \ref{sec_APL}. As $\{ \hat{\L}_t\}$ are stochastic, we use the average value $\bar{\L}$ over 1000 iterations for diagnosing convergence, which is computed after every 1000 iterations. The algorithm is terminated when the relative increase in $\bar{\L}$ is less than some tolerance (set as $10^{-5}$). The SVI algorithm with $E_{y|\theta^{(t)}}[s(y)]$ estimated using either option (a) Monte Carlo sampling or option (b) SNIS is described below. 

\floatname{algorithm}{SVI Algorithm: option (a) Monte Carlo sampling or option (b) SNIS}
\def\labelenumeratei{i.}
\begin{algorithm}[htb!]
\caption{}
\begin{enumerate}
\itemsep-0.1em 
\item Compute $\hat{\theta}_\ML$, $\log z(\hat{\theta}_\ML)$ and simulate $\mathcal{S}_{K_0} (\hat{\theta}_\ML)$. 
\item If option (b), simulate $\mathcal{S}_{K} (\hat{\theta}_\ML)$ and initialize $\mathbb{S} = \{\hat{\theta}_\ML, \mathcal{S}_{K} (\hat{\theta}_\ML)\}$.
\item Initialize $\mu^{(0)}$, $C^{(0)}$, $C'^{(0)}$, $\bar{\L}^{\text{old}}$ and $\epsilon=1$. Set $t = 0$.
\item While $\epsilon > \text{tolerance}$,
\vspace{-3mm}
\begin{enumerate}[i.]
\item $t \leftarrow t+1$.
\item Generate $s^{(t)} \sim N(0, I_p)$. Compute $\theta^{(t)} = C^{(t)} s^{(t)} + \mu^{(t)}$.
\item If option(a), simulate $\mathcal{S}_K(\theta^{(t)})$. $E_{y|\theta^{(t)}} \{s(y)\} \approx \frac{1}{K} \sum_{k=1}^K s(y_k^{(t)})$. \\ [1mm]
If option (b),
\begin{itemize}
\item Find $\theta_u \in \mathbb{S}$ closest in Mahalanobis distance to $\theta^{(t)}$.
\item Compute $w_k^u =\exp \{ s(y_k^u)^T (\theta^{(t)}-\theta^u) \}$, $\tilde{w}_k^u = w_k^u/ (\sum_{k'=1}^K w_{k'}^u )$ for $k=1, \dots, K$.
\item Compute $\ESS = 1/{\sum_{k=1}^K (\tilde{w}_k^u)^2}$.
\item If $\ESS \geq K/3$, $E_{y|\theta^{(t)}} \{s(y)\} \approx \sum_{k=1}^K  \tilde{w}_k^u s(y_k^u)$. 
\item If $\ESS < K/3$, simulate $\mathcal{S}_K(\theta^{(t)})$. $E_{y|\theta^{(t)}} \{s(y)\} \approx \frac{1}{K} \sum_{k=1}^K s(y_k^{(t)})$. \\[1mm]
$\mathbb{S} \leftarrow \mathbb{S} \cup\{ \theta^{(t)}, \mathcal{S}_K(\theta^{(t)})\}$. 
\end{itemize}
\item Compute $\hat{\nabla}\L_\mu^{(t)} = \nabla_\theta \log p(\theta^{(t)}, y) + {C^{(t)}}^{-T} s^{(t)}$. Update $\mu^{(t+1)}=\mu^{(t)} + \rho_t \hat{\nabla}\L_\mu^{(t)}$.
\item Update $\vech(C'^{(t+1)}) = \vech(C'^{(t)}) + \rho_t D_C\vech  \{\hat{\nabla}\L_\mu^{(t)} {s^{(t)}}^T\}.$ \\ [1mm]
Recover $C^{(t+1)}$ from $C'^{(t+1)}$.
\item Compute lower bound estimate $\hat{\L}_t$ using \eqref{LBest}. 
\item If $(t \equiv 0)$ mod 1000, $\bar{\L}^{\text{new}} = \frac{1}{1000} \sum_{i=t-999}^{i=t} \hat{\L}_i$, $\epsilon = (\bar{\L}^{\text{new}} - \bar{\L}^{\text{old}})/ |\bar{\L}^{\text{old}}|$, $\bar{\L}^{\text{old}} \leftarrow \bar{\L}^{\text{new}}$.
\end{enumerate}
\end{enumerate}
\vspace{-3mm}
\end{algorithm}

We recommend using an adaptive stepsize for $\{\rho_t\}$ in the SVI Algorithm, which adjusts to individual parameters and tends to lead to faster convergence. In our code, Adam \citep{Kingma2014b} is used for computing the stepsize and the tuning parameters are set close to recommended default values. We initialize $\mu$ and $C$ using estimates obtained from the NCVMP algorithm.

\section{Model selection} \label{sec:Model_sel}
A challenging aspect of fitting ERGMs is determining which sufficient statistics to include in the model. In Bayesian inference, different models can be compared using the Bayes factor \citep{Kass1995}. Let $\M_1, \dots, \M_R$ be candidate models for data $y$, with respective parameters, $\theta_1, \dots, \theta_R$, and prior probabilities, $p(\M_1), \dots, p(\M_R)$. Under prior densities $p(\theta_1|\M_1)$, \dots, $p(\theta_R|\M_R)$ of the parameters, the marginal likelihood of $y$ is 
\begin{equation*}
p(y|\M_r) = \int p(y|\theta_r, \M_r) p(\theta_r|\M_r) d\theta_r, \quad r=1, \dots, R.
\end{equation*}
We can then compare the models in terms of their posterior probabilities,
\begin{equation*}
p(\M_r|y) = \frac{p(y|\M_r)p(\M_r)}{\sum_{r=1}^{R} p(y|\M_r)p(\M_r)}, \quad r=1, \dots, R.
\end{equation*}
If each model is a priori equally likely, $p(\M_r)=1/R$ for $r=1, \dots, R$, then selecting the model with the highest posterior probability, $\argmax_r p(\M_r|y)$ is the same as selecting the one with the highest marginal likelihood, $\argmax_r p(y|\M_r)$. The Bayes factor for comparing any two models $\M_r$ and $\M_j$ is $\BF_{rj} = {p(y|\M_r)} /{p(y|\M_j)}$. When $\BF_{rj} > 1$, the data favors $\M_r$ over $\M_j$. If $0<\BF_{rj} < 1$, $\M_j$ is favored instead. To compare more than two models at the same time, we can choose one model as the reference and compute the Bayes factors relative to that reference.

While $\L$ provides a lower bound to the log marginal likelihood $\log p(y)$, using $\L$ for model selection may not yield reliable results as the tightness of the bound cannot be quantified. Instead, we consider the {\em importance weighted lower bound} \cite[IWLB,][]{Burda2016}. Let $\theta_1, \dots, \theta_V$ be generated independently from $q_\lambda(\theta)$. From Jensen's inequality,
\begin{equation*}
\log p(y) = \log \int \frac{p(y, \theta)}{q_\lambda(\theta)} q_\lambda(\theta) d\theta 
= \log E_{q_\lambda} \bigg( \frac{1}{V} \sum_{v=1}^V \omega_v \bigg)
\geq  E_{q_\lambda} \bigg( \log \frac{1}{V} \sum_{v=1}^V \omega_v \bigg) = \L_V^\IW,
\end{equation*}
 where $\omega_v = p(y, \theta_v)/ q_\lambda(\theta_v)$. Thus $\L_V^\IW$ provides a lower bound to $\log p(y)$. \cite{Burda2016} show that $\L_V^\IW$ increases monotonically with $V$ and approaches $\log p(y)$ as $V \rightarrow \infty$ (by the strong law of large numbers). Thus $ \L_V^\IW$ is an asymptotically unbiased estimator of $\log p(y)$, which can be useful in model selection if $V$ is sufficiently large. 

We use the IWLB algorithm to compute an IWLB estimate from the fitted $q_\lambda(\theta)$, where the importance weights are approximated by (I) plugging in the adjusted pseudolikelihood for the true likelihood or (II) a Monte Carlo estimate based on $\hat{\theta}_\ML$ as reference. The algorithm increases $V$ by $J$ units in each iteration until the increase in the IWLB is negligible. We set the tolerance as $10^{-5}$, $N=1000$ and each increment $J=50$.

\floatname{algorithm}{IWLB Algorithm}
\begin{algorithm}[htb!]
\caption{}
\begin{enumerate}
\itemsep-0.1em 
\item Initialize $t=0$, $V=0$, $\epsilon=1$, $\bar{\omega}$ as a zero vector of length $N$ and $(\hat{\L}_V^{\IW})^{\text{old}}= \L$. 
\item While $\epsilon  >$ tolerance,
\vspace{-2mm}
\begin{enumerate}[i.]
\item $t \leftarrow t + 1$, $ V \leftarrow V + J$.
\item Generate $NJ$ samples $\{\theta_{11}, \dots, \theta_{NJ} \}$ from $q_\lambda(\theta)$.
\item For $i=1, \dots, N$, $j=1, \dots, J$, $\log \omega_{ij} = \log p(y|\theta_{ij}) + \log p(\theta_{ij}) - \log q_\lambda(\theta_{ij})$, where 
\begin{equation*}
\log p(y|\theta_{ij}) \approx \begin{cases}
\log  \tilde{p}(y, \theta_{ij}) & \text{if (I)},\\
\theta_{ij}^T s(y) - \log z(\hat{\theta}_\ML) - \log \big[ \frac{1}{K_0} \sum_{k=1}^{K_0} \e^{s(y_k)^T (\theta_{ij}-\hat{\theta}_\ML)} \big] & \text{if (II)}.
\end{cases}
\end{equation*}
\item Compute $\tilde{\omega} = [\tilde{\omega}_1, \dots, \tilde{\omega}_N]^T$ where $\tilde{\omega}_i = \sum_{j=1}^J \omega_{ij}$.
\item $\bar{\omega} \leftarrow \{(t-1)J\bar{\omega}  + \tilde{\omega}\}/V$.
\item Compute IWLB estimate: $(\hat{\L}_V^{\IW})^{\text{new}} = \frac{1}{N} \sum_{i=1}^N \log \bar{\omega}_i $.
\item Compute $\epsilon = \{ (\hat{\L}_V^{\IW})^{\text{new}} - (\hat{\L}_V^{\IW})^{\text{old}} \}/|(\hat{\L}_V^{\IW})^{\text{old}} |$. 
$(\hat{\L}_V^{\IW})^{\text{old}} \leftarrow  (\hat{\L}_V^{\IW})^{\text{new}}$.
\end{enumerate}
\end{enumerate}
\vspace{-3mm}
\end{algorithm}
The computational advantage in using the IWLB for model selection as compared to MCMC approaches is that it does not require tuning, determining appropriate length of burn-in or checking of convergence diagnostics, which can be tedious when a large number of candidate models are compared.

\section{Applications} \label{sec:Appl}
We illustrate the performance of proposed variational methods using three real networks shown in Figure \ref{networks}. The code for the variational algorithms is written in Julia 0.6.4 and the experiments are run on a Intel Core i9 CPU @ 3.60GHz, 16.0GB RAM. Maximum likelihood estimation, maximum pseudolikelihood estimation and simulation of networks from the likelihood are performed using the {\tt ergm} R package version 3.8.0 \citep{Hunter2008}. This is done in Julia using the {\tt RCall} package. We use the {\tt bergm} function from the {\tt Bergm} R package version 4.1.0 \citep{Caimo2014} to sample from the posterior distribution via the exchange algorithm \citep{Caimo2011}. The number of chains must be at least the dimension of $\theta$ and adaptive direction sampling (ADS) is used to improve mixing. We use the default in {\tt bergm}, which sets the number of chains to be twice the length of $\theta$. The parameters $\gamma$ (move factor in ADS) and $\sigma_\epsilon$ (variance of normal proposal) are tuned so that the acceptance rate lies between 20\% and 25\%. Posterior distributions estimated using the exchange algorithm are regarded as ``ground truth" and we use them to evaluate the accuracy of posterior distributions approximated using variational methods. We also compare the variational densities with Gaussian posteriors obtained using {\em Laplace approximation}, where the adjusted pseudolikelihood is plugged in for the true likelihood. Details are given in the supplementary material. 
\begin{figure}[tb!]
\centering
\includegraphics[width=0.325\textwidth]{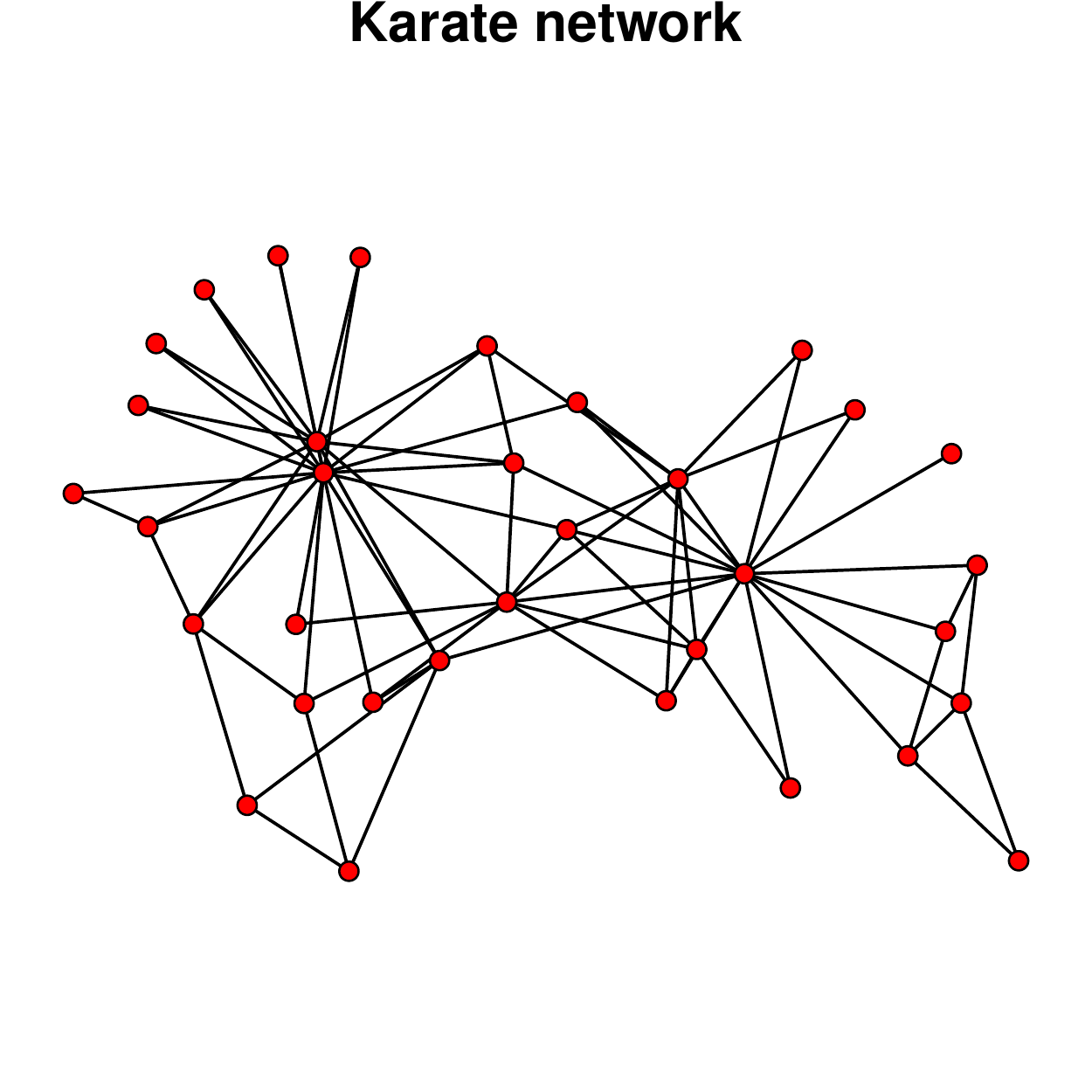}
\includegraphics[width=0.325\textwidth]{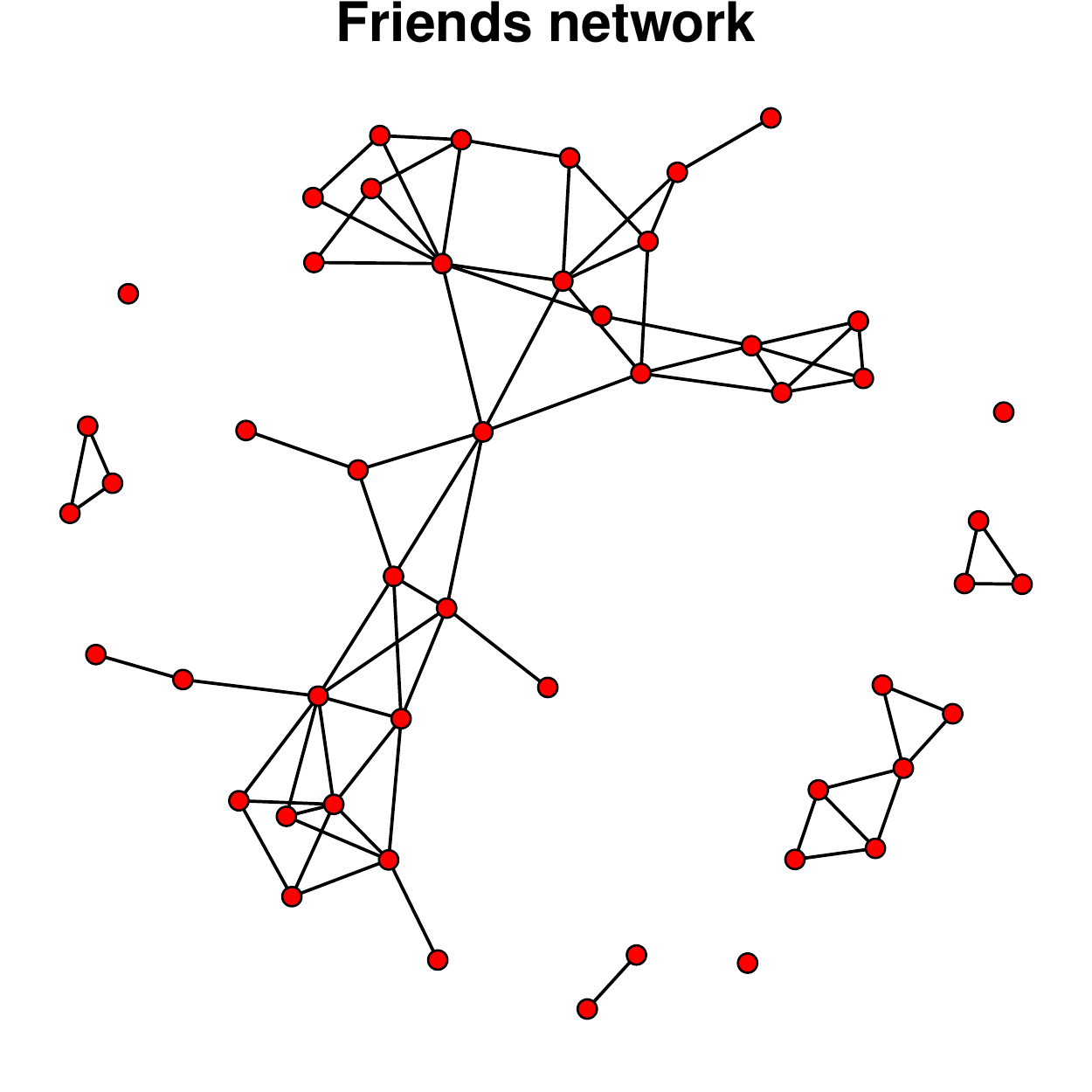}
\includegraphics[width=0.33\textwidth]{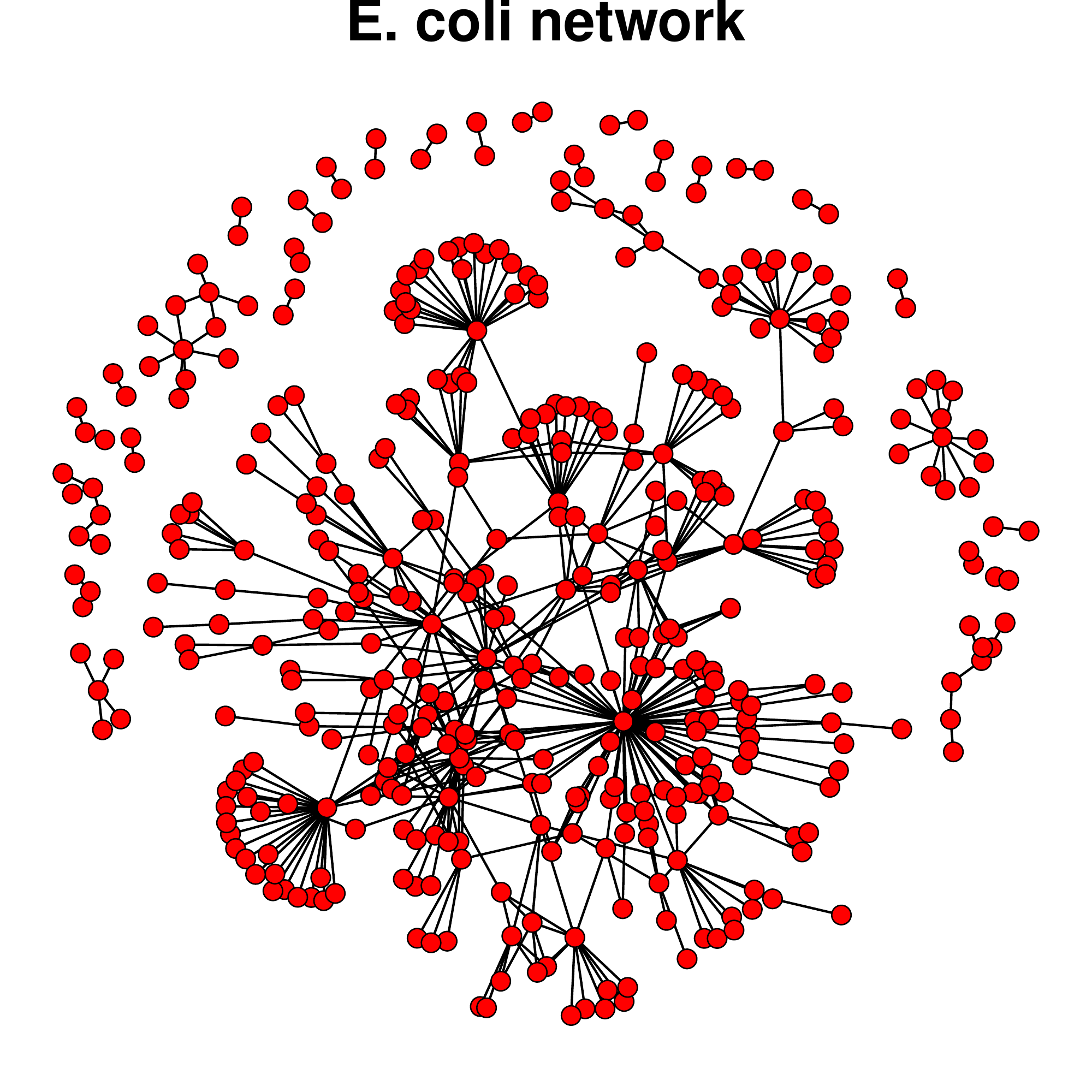}
\caption{\small Plots of karate, friends and ecoli networks.} \label{networks}
\end{figure}

An estimate of the marginal likelihood based on the adjusted pseudolikelihood using Chib and Jeliazkov's method \citep[CJ, ][]{Chib2001} is obtained using the {\tt evidence\_CJ} function \citep{Bouranis2018} from {\tt Bergm}. The parameter for the Metropolis sampling is tuned such that the acceptance rate lies between 20\% and 25\%. We set the total number of iterations as 25,000 including a burn-in of 5000 in each case. All three methods (CJ, NCVMP and Laplace) rely on the adjusted pseudolikelihood; CJ's method samples from the posterior while the latter two use Gaussian approximations. NCVMP tries to minimize the KL divergence between the Gaussian approximation and the true posterior, while the normal density in Laplace approximation is centered at the posterior mode with the covariance matrix taken as the negative inverse Hessian of $\log p(y, \theta)$ evaluated at the mode. 

While updates in NCVMP are deterministic, the SVI algorithm is subject to random variation. We consider $K \in \{1,5,20\}$ for (a) Monte Carlo sampling and $K \in \{100, 200, 500\}$ for (b) SNIS, and study the performance of each setting using ten runs from different random seeds. The {\tt kullback\_leibler\_distance} function from the {\tt philentropy} R package is used to compute the KL divergence of the marginal posterior of $\{\theta_j\}$ obtained using an approximation method from that estimated using the exchange algorithm. We set $\mu_0 = 0$ and $\Sigma_0 = 100 I_p$ for the prior distribution of $\theta$ throughout.

In later examples, we consider models which contain the following sufficient statistics concerning network structure. The first, $s_L(y) = \sum_{i < j} y_{ij}$ is the number of edges, which accounts for the overall density of the observed network. We also consider
\begin{equation*}
\begin{aligned}
 s_{\gwd}(y, \phi_u) &= \e^{\phi_u} \sum_{\ell=1}^{n-1} \{ 1 - (1 - \e^{-\phi_u})^\ell\}  \text{D}_\ell (y), \\
 s_{\gwesp}(y, \phi_v) &= \e^{\phi_v} \sum_{\ell=1}^{n-2} \{ 1 - (1 - \e^{-\phi_v})^\ell \}  \text{EP}_\ell (y),
\end{aligned}
\end{equation*}
which are respectively the geometrically weighted degree (gwd) statistic for modeling the degree distribution of the network and the geometrically weighted edgewise shared partners (gwesp) statistic for modeling transitivity. D$_\ell(y)$ counts the number of nodes in $y$ that have $\ell$ neighbors, and EP$_\ell(y)$ counts the number of connected dyads in $y$ that have exactly $\ell$ common neighbors. These two statistics improve the fit of ERGMs by placing geometrically decreasing weights on higher order terms \citep{Hunter2008a}.

\subsection{Karate network} \label{sec:karate}
The karate club network \citep{Zachary1977} contains 78 undirected friendship links among 24 members, constructed based on interactions outside club activities. This data is available at \url{https://networkdata.ics.uci.edu/}. We consider three competing models \citep{Caimo2014, Bouranis2018}, whose unnormalized likelihoods are 
\begin{equation}\label{models}
\begin{aligned}
\mathcal{M}_1 &: \exp \{\theta_{1} s_L(y) + \theta_{2} s_{\gwesp}(y, 0.2)\}, \\
\mathcal{M}_2 &: \exp \{\theta_{1} s_L(y) + \theta_{2} s_{\gwd}(y,0.8)\}, \\
\mathcal{M}_3 &: \exp \{\theta_{1} s_L(y) + \theta_{2} s_{\gwesp}(y, 0.2)+ \theta_3 s_{\gwd}(y, 0.8)\}.
\end{aligned}
\end{equation}

First, we estimate parameters in the adjusted pseudolikelihood using the {\tt ergm} and {\tt simulate} functions in the {\tt ergm} R package. We set the number of auxiliary iterations as 30000, thinning factor as 1000 and number of simulations as 1000. Computation times are 6.3, 2.4 and 6.9 seconds for $\M_1$, $\M_2$ and $\M_3$ respectively. Next, we fit the three models using CJ's method, Laplace approximation and the exchange and variational algorithms. For the exchange algorithm, the number of auxiliary iterations is also set as 30000, length of burn-in as 1000 and number of iterations per chain excluding burn-in as 10000. We have 4 chains for $\M_1$ and $\M_2$ (40000 samples), and 6 chains for $\M_3$ (60000 samples). Setting the ADS move factor $\gamma$ as 1.1, 1.175 and 0.775 for $\M_1$, $\M_2$ and $\M_3$ in order, the average acceptance rates are 22.5\%, 22.8\% and 23.0\%. For CJ's method, the acceptance rates for $\M_1$, $\M_2$ and $\M_3$ after tuning are 21.7\%, 21.9\% and 23.9\% respectively. 

Computation times of the various algorithms are shown in Table \ref{karate_computation_times}. 
\begin{table}[htb!]
\centering
\begin{small}
\begin{tabular}{ll|cccl}
\hline
Method & $K$ & \multicolumn{3}{c}{Computation time (seconds)}  \\
&& $\mathcal{M}_1$ &  $\mathcal{M}_2$ & $\mathcal{M}_3$ \\
\hline
Exchange && 770.7 & 258.8 & 1264.4 \\ \hline
CJ  &&  3.9 & 4.0 & 6.1  \\
Laplace && 0.1 & 0.0 & 0.0 \\
NCVMP && 0.4 & 0.1 & 0.1  \\ \hline
\multirow{3}{*}{SVI (a)} & 1&  58.7 $\pm$ 14.6 & 79.0 $\pm$ 7.2 & 68.7 $\pm$ 19.2 \\
& 5 &   57.5 $\pm$ 12.3 & 80.2 $\pm$ 11.5 & 83.6 $\pm$ 26.5 \\
& 20 &   90.5 $\pm$ 20.4 & 102.0 $\pm$ 14.8 & 109.1 $\pm$ 38.1 \\ \hline
\multirow{3}{*}{SVI (b)} &  100 &  3.0 $\pm$ 0.2 (23, 29) & 2.4 $\pm$ 0.2 (29, 39) & 10.5 $\pm$ 1.3 (89, 142) \\
& 200 & 4.4 $\pm$ 0.4 (23, 30) & 3.2 $\pm$ 0.2 (34, 40) & 18.4 $\pm$ 2.1 (94, 149) \\
& 500 & 9.0 $\pm$ 0.9 (23, 31) & 5.3 $\pm$ 0.4 (33, 42) & 40.0 $\pm$ 5.6 (98, 159) \\
\hline
\end{tabular}
\end{small}
\caption{\small Karate network. Computation times of various algorithms. For SVI (b), the range of the number of particles in $\mathbb{S}$ over ten runs is given in brackets.} \label{karate_computation_times}
\end{table}
Methods relying on the adjusted pseudolikelihood are the fastest. Among these, the algorithms with deterministic updates (Laplace and NCVMP) are faster than CJ's method which samples from the posterior. SVI (b) converges much faster than SVI (a) for the 2-dimensional models but the speedup is reduced in the 3-dimensional case. This is likely due to the higher rate of sampling from the likelihood, since the number of particles in $\mathbb{S}$ for $\mathcal{M}_3$ is about 3--4 times larger than that for $\mathcal{M}_1$ and $\mathcal{M}_2$. 

The KL divergence of the approximate marginal posterior of each $\theta_j$ from that estimated using the exchange algorithm is shown in Figure \ref{karate_KLdiv}. 
\begin{figure}[tb!]
\centering
\includegraphics[width=\textwidth]{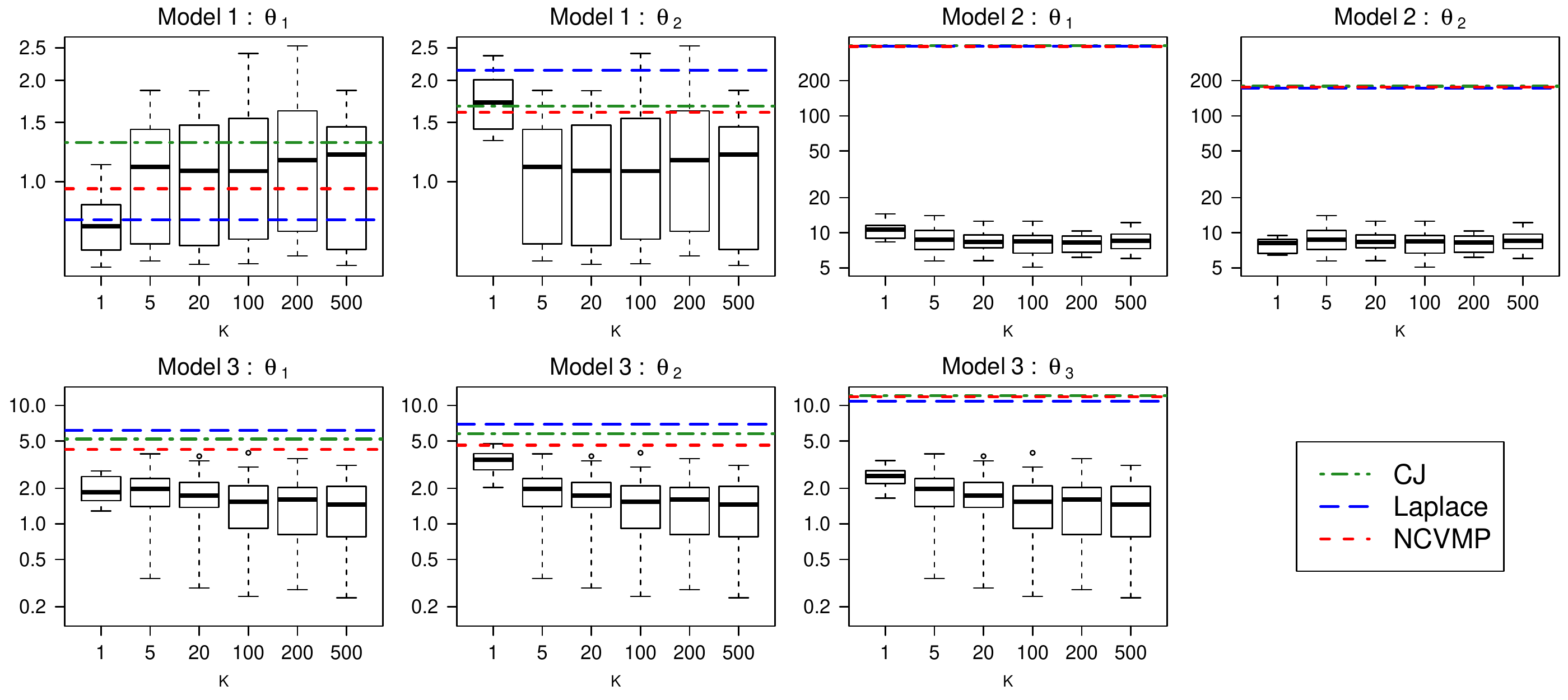}
\caption{\small Karate network. KL divergence of approximate marginal posterior of each $\theta_j$ from posterior estimated using exchange algorithm. In each figure, the first three and last three boxplots correspond to SVI (a) and SVI (b) respectively.} \label{karate_KLdiv}
\end{figure}
NCVMP performs better than CJ's method in all cases, while Laplace is sometimes better and sometimes worse than NCVMP. For SVI (a), the performance of $K=1$ appears to fluctuate more than $K\in \{5,20\}$. However, $K=5$ and $K=20$ perform similarly, and it seems sufficient to use $K=5$. A single sample may not be able to provide sufficient gradient information and stability and we recommend using a few samples for averaging. Performance of SVI (b) is similar across $K \in \{100, 200, 500\}$ and is close to SVI (a) for $K\in \{5,20\}$. The methods relying on adjusted pseudolikelihood (CJ, Laplace, NCVMP) perform worse than the SVI algorithms for $\M_3$ and especially $\M_2$. 

The marginal posteriors of each $\theta_j$ obtained from the exchange algorithm, NCVMP and one randomly selected run of the SVI algorithms are shown in Figure \ref{karate_posterior}. Marginal posteriors from CJ's method and Laplace approximation are not shown to avoid clutter, but they are almost identical to that of NCVMP.
\begin{figure}[htb!]
\centering
\includegraphics[width=0.9\textwidth]{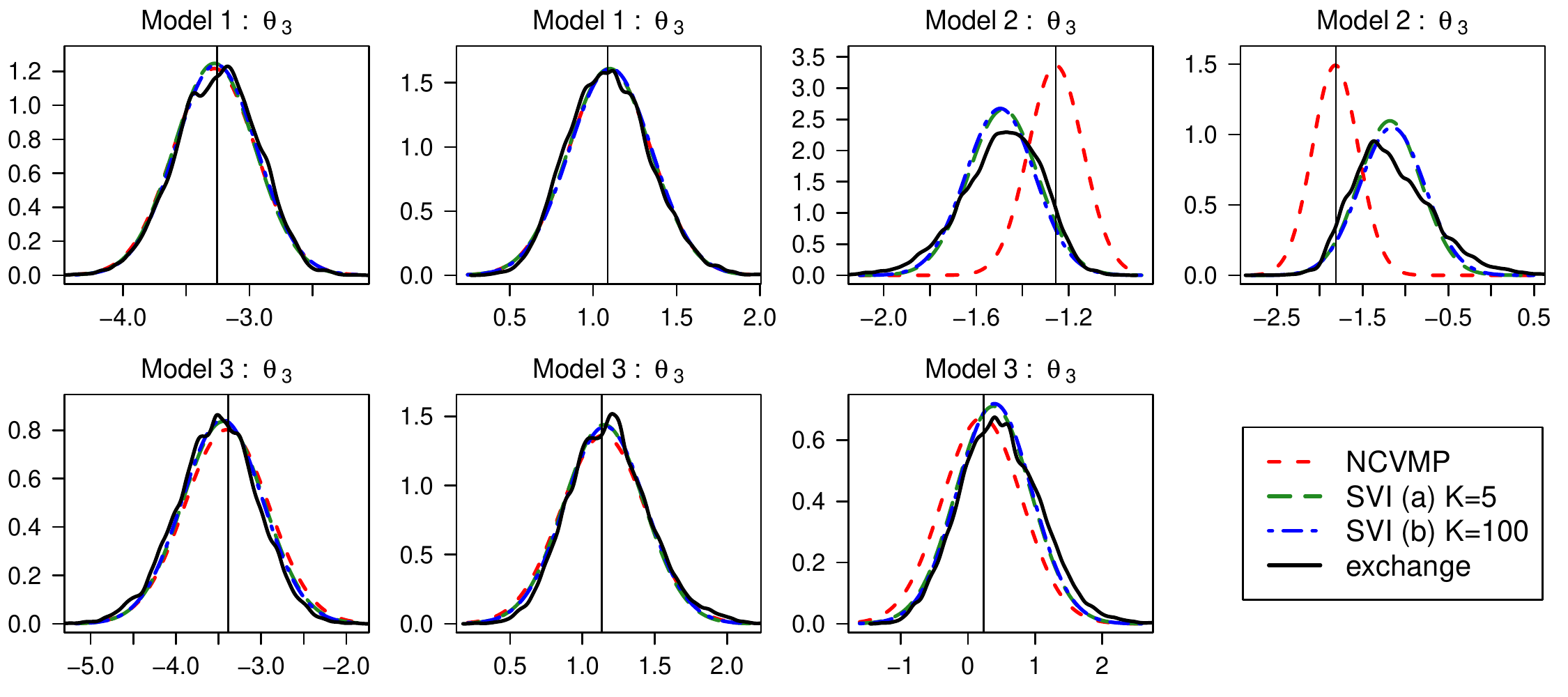}
\caption{\small Karate network. Marginal posterior distributions of each $\theta_j$ from variational and exchange algorithms. Vertical lines indicate values of $\hat{\theta}_\ML$.} \label{karate_posterior}
\end{figure}
For $\M_1$ and $\M_3$, the marginal posteriors obtained using different approaches are quite similar. However, for $\M_2$, the posteriors obtained using NCVMP are quite different from the posteriors obtained using the SVI and exchange algorithms. The difference is likely due to the adjusted pseudolikelihood not being able to mimic the true likelihood well. In particular, the posterior means estimated using NCVMP remain very close to $\hat{\theta}_\ML$. The SVI algorithms capture the (slightly skewed) true posteriors better much than NCVMP for $\M_2$. 

Figure \ref{karate_ESS_particles} shows the ESS and particles in $\mathbb{S}$ 
\begin{figure}[htb!]
\centering
\includegraphics[width=\textwidth]{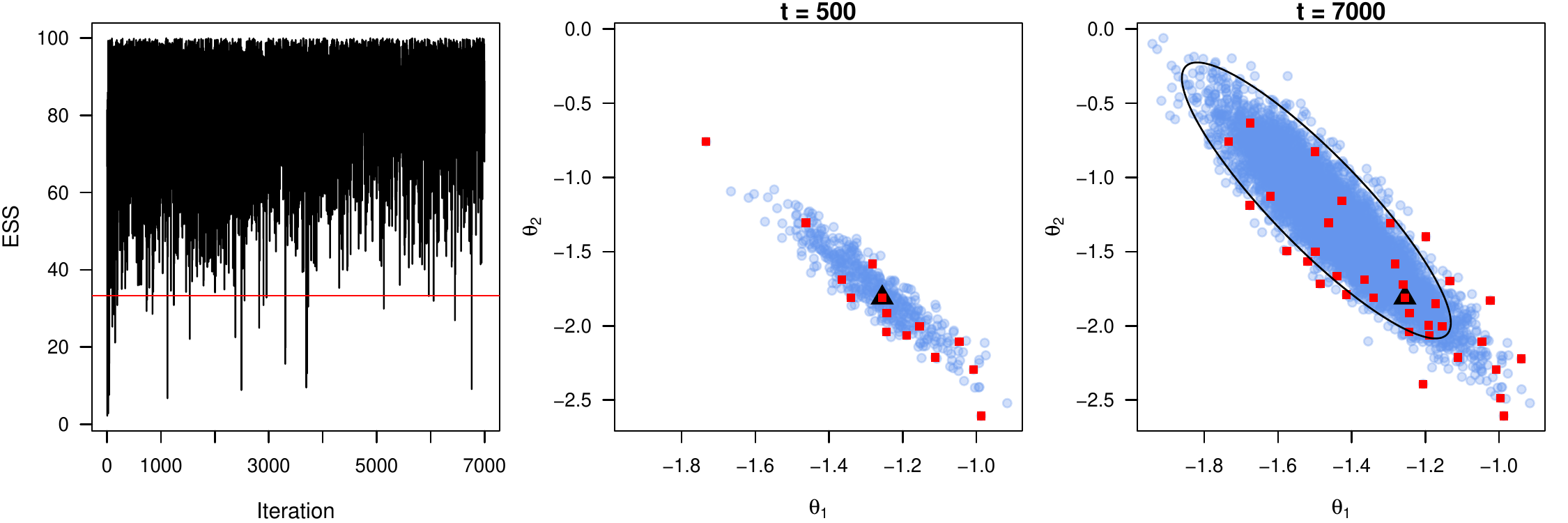}
\caption{\small Karate network. SVI (b) ($K=100$) was run for $\M_2$. Left plot shows ESS at each iteration with red line marking the threshold of $K/3$. Middle and right plot shows $\{\theta^{(t)}\}$ in blue circles and the particles in $\mathbb{S}$ in red squares at the 500th and 7000th iteration respectively. Maximum likelihood estimate is marked using a black triangle. Ellipse in black is the normal density contour that contain 95\% of the probability of the estimated $q_\lambda(\theta)$.} \label{karate_ESS_particles}
\end{figure}
when SVI (b) was run for $\M_2$ with $K=100$. The algorithm converged in 7000 iterations and there were 35 particles in $\mathbb{S}$ eventually. The leftmost plot indicates that the ESS falls below the threshold of $K/3$ more frequently at the beginning. Thus the first part of the iterations are generally more-consuming due to simulation from the likelihood as particles are being added to $\mathbb{S}$. At the 500th iteration, $\{\theta^{(t)}\}$ and the particles in $\mathbb{S}$ are still centered around $\hat{\theta}_{\ML}$. However, the algorithm eventually moves away from $\hat{\theta}_{\ML}$ towards the true posterior mean. The particles in $\mathbb{S}$ are quite evenly spread out across the estimated $q_\lambda(\theta)$.
 
Using the IWLB algorithm, we computed IWLB estimates using (I) adjusted pseudolikelihood for Laplace approximation and NCVMP algorithm, and (II) Monte Carlo estimate for the SVI algorithms (see Table \ref{karate_IWLB}).   
\begin{table}[htb!]
\centering
\begin{small}
\begin{tabular}{ll|ccc|ccc}
\hline
Method & $K$ & \multicolumn{3}{c}{Computation time (seconds)}  & \multicolumn{3}{c}{V} \\
&& $\mathcal{M}_1$ &  $\mathcal{M}_2$ & $\mathcal{M}_3$ & $\mathcal{M}_1$ &  $\mathcal{M}_2$ & $\mathcal{M}_3$\\
\hline
CJ && -219.3 &  -232.6 & -221.8 \\
Laplace & & -219.3 (0.6) & -232.6 (0.4) & -221.8 (1.3) & 100 & 50 & 100 \\
NCVMP & & -219.3 (0.2) & -232.6 (0.3) & -221.8 (1.5) & 50 & 50 & 100 \\ \hline
\multirow{3}{*}{SVI (a)}  & 1 & -219.4 (2.7) & -231.2 (4.9) & -221.7 (4.4) & 100 & (150, 250) & (150, 250) \\
& 5 & -219.4 (2.7) & -231.2 (4.3) & -221.7 (4.2) & 100 & (150, 250) & (150, 200) \\
& 20 & -219.4 (2.7) & -231.2 (4.4) & -221.7 (4.2) & 100 & (150, 250) & (150, 200) \\ \hline
\multirow{3}{*}{SVI (b)} & 100 & -219.4 (2.7) & -231.2 (4.3) & -221.7 (4.1) & 100 & (150, 250) & 150 \\
& 200 & -219.4 (2.7) & -231.2 (4.1) & -221.7 (4.1) & 100 & 150 & 150 \\
& 500 & -219.4 (2.7) & -231.2 (4.4) & -221.7 (4.1) & 100 & (150, 250) & 150 \\
\hline
\end{tabular}
\end{small}
\caption{\small Karate network. IWLB computed using (I) adjusted pseudolikelihood for Laplace and NCVMP and (II) Monte Carlo for SVI algorithms. Computation times (in brackets) and range of values of $V$.} \label{karate_IWLB}
\end{table}
Results for SVI algorithms are averaged over ten runs but the standard deviations are almost zero and hence only the means are displayed. The results from Laplace and NCVMP are identical to CJ's method (all three methods are based on adjusted pseudolikelihood). For the SVI algorithms, the results for $\M_1$ and $\M_3$ are very close to CJ's method but the results for $\M_2$ differ slightly. Recall that methods based on the adjusted pseudolikelihood were unable to  approximate the true posterior accurately for $\M_2$. Thus the IWLB estimate from the SVI algorithms may be more reliable than CJ's estimate for $\M_2$. Finally, it is clear that $\M_1$ is the most favored model, followed by $\M_3$ and then $\M_2$.

 \subsection{Teenage friends and lifestyle study}
Here we consider a subset of 50 girls from the ``Teenage friends and lifestyle study" data set \citep{Pearson2000} available at \url{https://www.stats.ox.ac.uk/~snijders/siena/s50_data.htm}. In this study, friendship links among the students were recorded over three years from 1995 to 1997. The students were also surveyed about their smoking behavior and frequency of drugs consumption among other lifestyle choices. We consider the friendship network at the first time point and the qualitative attributes {\tt smoke}  (1: non-smoker, 2: occasional smoker or 3: regular smoker), {\tt drugs} (1: non-drug user or tried once or 2: occasional or regular drug user) and {\tt sport} (1: not regular or 2: regular). Figure \ref{friendsplot} shows plots of the friendship network according to the attributes. 

There appears to be some homophily in friendships by smoking and drug usage behavior as nodes of the same color (attribute value) seem to have a higher tendency to form links. We consider three models, whose unnormalized likelihoods are given by
\begin{equation*}
\begin{aligned}
\mathcal{M}_1 &: \exp \{\theta_{1} s_L(y) + \theta_{2} s_\gwesp(y, \log 2) + \theta_{2} s_\gwd(y, 0.8)\}, \\
\mathcal{M}_2 &: \exp \{\theta_{1} s_L(y) + \theta_{2} s_\gwesp(y, \log 2) + \theta_{2} s_\gwd(y, 0.8) + s_{\text{drugs}}(y) \}. \\
\mathcal{M}_3 &: \exp \{\theta_{1} s_L(y) + \theta_{2} s_\gwesp(y, \log 2) + \theta_{2} s_\gwd(y, 0.8) + s_{\text{smoke}}(y) + s_{\text{drugs}}(y) + s_{\text{sport}}(y) \}.
\end{aligned}
\end{equation*}
Here $s_{\text{smoke}}(y)$, $s_{\text{drugs}}(y)$ and $s_{\text{sport}}(y)$ counts the number of connected dyads $(i,j)$ for which nodes $i$ and $j$ have the same value for the attributes {\tt smoke}, {\tt drugs} and {\tt sport} respectively. In the {\tt ergm} R package, these terms are coded as {\tt nodematch(`smoke'))}, {\tt nodematch(`drugs')} and {\tt nodematch(`sport')}. 
\begin{figure}[htb!]
\centering
\includegraphics[width=0.9\textwidth]{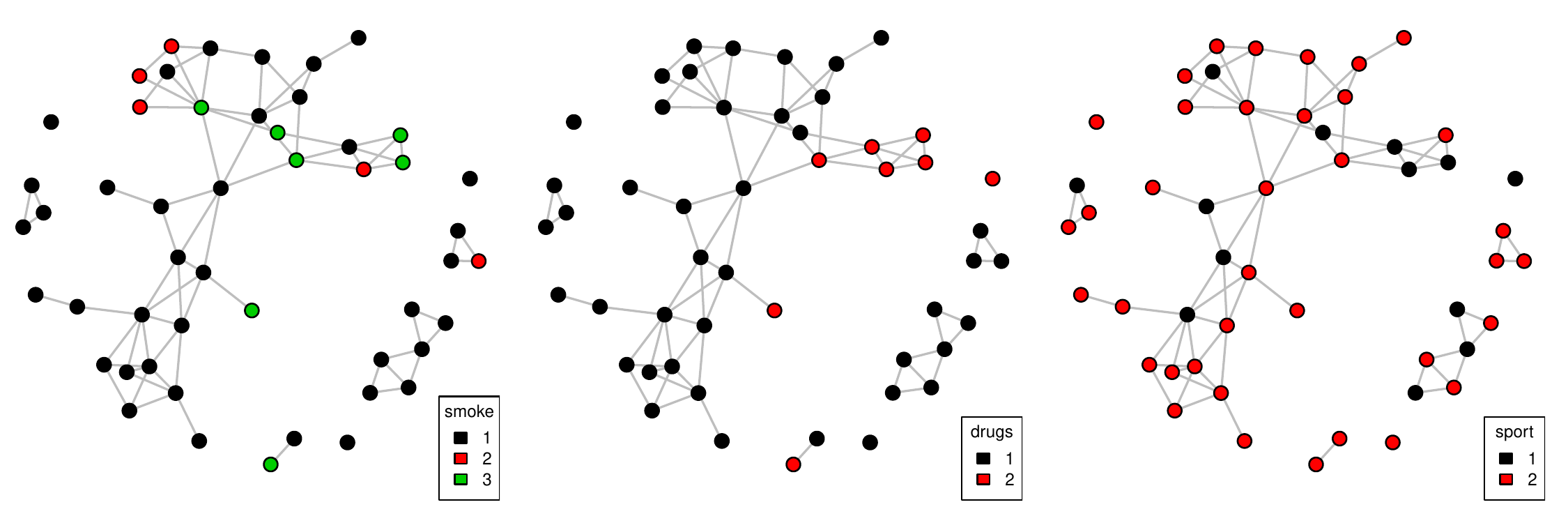}
\caption{Plot of friends network, where nodes are colored according to the attributes.} \label{friendsplot}
\end{figure}

Estimating the adjusted pseudolikelihood took 4.9 and 5.3 and 6.1 seconds for $\M_1$, $\M_2$ and $\M_3$ respectively. We set the number of auxiliary iterations as 50000, thinning factor as 1000 and number of simulations as 1000. Next we fit the three models using CJ's method, Laplace approximation and the exchange and variational algorithms. For the exchange algorithm, the number of auxiliary iterations is also set as 50000, the length of burn-in as 1000 and the number of iterations for each chain excluding burn-in to be 10000. We have 6 chains for $\M_1$ (60000 samples), 8 chains for $\M_2$ (80000 samples) and 12 chains for $\M_3$ (120000 samples). The ADS move factor $\gamma$ is adjusted as 1.8, 1.45 and 1.1 so that average acceptance rates are 22.7\%, 22.2\% and 22.3\% for $\M_1$, $\M_2$ and $\M_3$ respectively. For CJ's method, the tuning parameters were adjusted such that the acceptance rates are 22.1\%, 22.2\% and 23.2\% for $\M_1$  $\M_2$ and $\M_3$ respectively. The computation times for various algorithms are shown in Table \ref{friends_computation_times}. 
\begin{table}[htb!]
\centering
\begin{small}
\begin{tabular}{ll|cccl}
\hline
Method & $K$ & \multicolumn{3}{c}{Computation time (seconds)}  \\
&& $\mathcal{M}_1$ &  $\mathcal{M}_2$ & $\mathcal{M}_3$ \\ \hline
Exchange && 1778.8 & 2290.3 & 3612.9 \\ \hline
CJ &&   5.1 & 5.8 & 7.9 \\
Laplace && 0.0 & 0.0 & 0.0 \\
NCVMP && 0.1 & 0.1 & 0.2 \\ \hline
\multirow{3}{*}{SVI (a)} & 1&  83.2  $\pm$ 25.9 & 93.3 $\pm$ 26.4 & 116.9  $\pm$  39.6 \\
& 5 & 106.1  $\pm$  20.3 & 102.2  $\pm$ 25.3 & 102.7  $\pm$  11.2 \\
& 20 & 118.6  $\pm$  33.9 & 131.7 $\pm$ 41.5 & 136.2  $\pm$  44.1 \\ \hline
\multirow{3}{*}{SVI (b)} & 100 & 10.9 $\pm$ 1.4 (81, 124) & 22.8 $\pm$ 4.3 (197, 332) & 86.0 $\pm$ 15.0 (895, 1444) \\
& 200 & 18.2 $\pm$ 2.9 (89, 139) & 40.4 $\pm$ 7.9 (218, 381) & 136.3 $\pm$ 11.6 (758, 1045) \\
& 500 & 40.0 $\pm$ 6.3 (98, 141) & 88.6 $\pm$ 18.1 (230, 378) & 332.2 $\pm$ 57.5 (1062, 1674) \\ \hline
\end{tabular}
\end{small}
\caption{\small Friends network. Computation times of various algorithms. For SVI (b), the range of the number of particles in $\mathbb{S}$ over ten runs is given in brackets.} \label{friends_computation_times}
\end{table}
Laplace is the fastest among methods relying on the adjusted pseudolikelihood, followed by NCVMP and CJ's method. For SVI (a), there is a gradual increase in computation times with the dimension of $\theta$. However, for SVI (b), the computation times increase quite sharply. Runtime for $\M_2$ ($p=4$) is about twice that of $\M_1$ ($p=3$) and runtime for $\M_3$ ($p=6$) is about eight times that of $\M_1$ ($p=3$). This is due to the high rate of simulating from the likelihood as can be seen from the drastic increase in the number of particles in $\mathbb{S}$. This phenomenon is due to the curse of dimensionality; a large number of particles are required to cover the region in the parameter space where $q_\lambda(\theta)$ is practically non-zero when $\theta$ is high-dimensional. While there is a clear advantage in using SVI (b) when $p \leq 4$, SVI (a) may be computationally more efficient for $p > 6$. Figure \eqref{friends_ESS} shows the ESS for one instance of  
SVI (b) when $K=100$. 
\begin{figure}[htb!]
\centering
\includegraphics[width=0.9\textwidth]{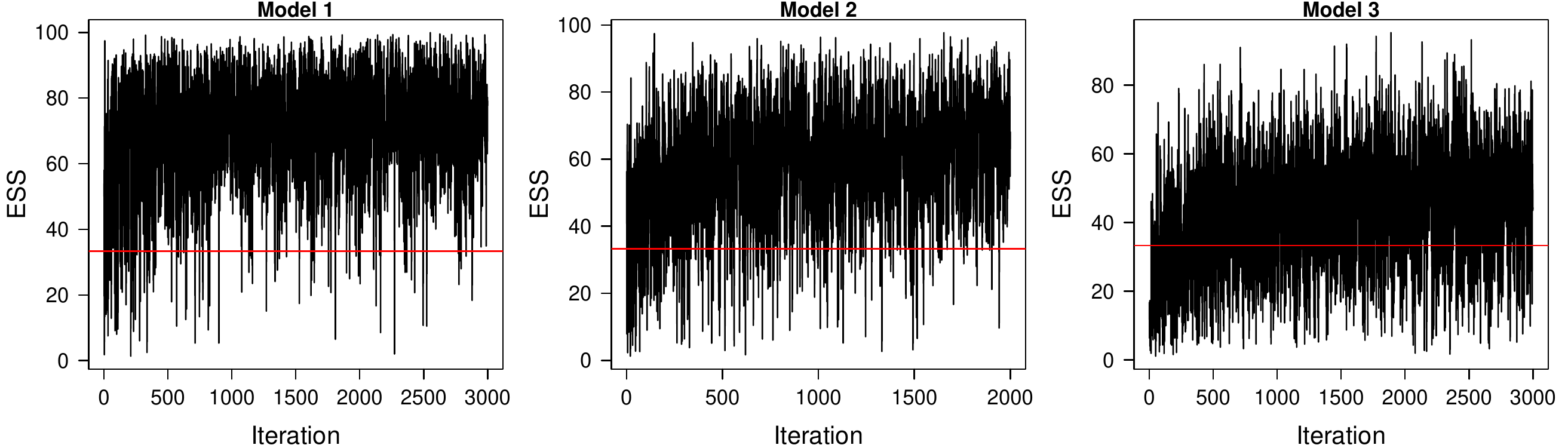}
\caption{\small Friends network. ESS at each iteration for a run of SVI (b) ($K=100$). Red horizontal line indicates threshold of $K/3$.} \label{friends_ESS}
\end{figure}
The number of particles in $\mathbb{S}$ are 112, 217 and 963 respectively. The ESS for $\M_3$ stills falls below the threshold at a high frequency even when the algorithm is close to convergence. This is because, when the parameter space is high dimensional, there is a high probability that a $\theta^{(t)}$ which is far away from any existing particles in $\mathbb{S}$ is generated.

Figure \ref{friends_KLdiv} compares the accuracy of the approximate marginal posterior of each $\theta_j$ relative to that estimated using the exchange algorithm.
\begin{figure}[htb!]
\centering
\includegraphics[width=\textwidth]{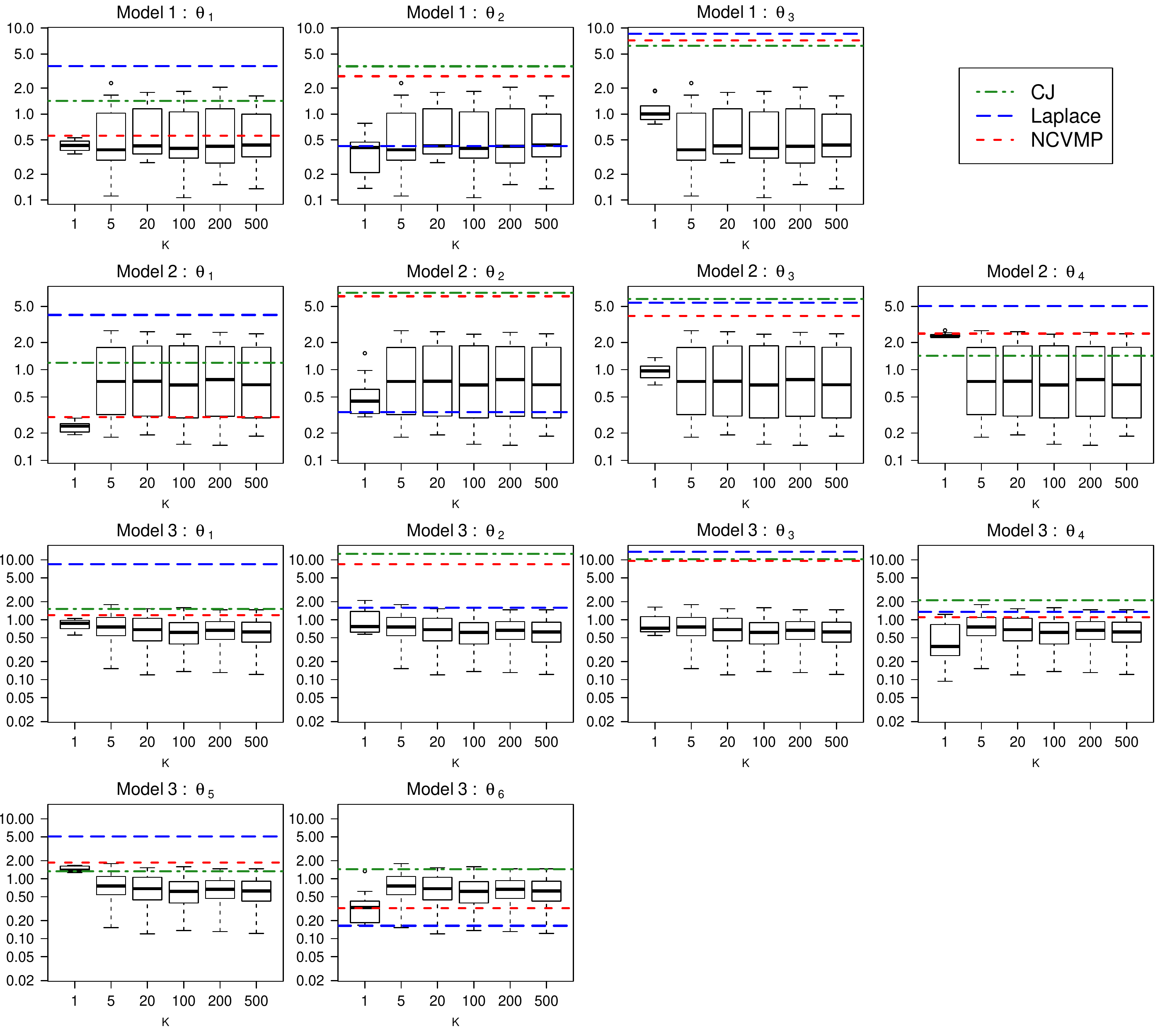}
\caption{\small Friends network. KL divergence of approximate marginal posterior of each $\theta_j$ from posterior estimated using exchange algorithm. In each figure, the first three and last three boxplots correspond to SVI (a) and SVI (b) respectively.} \label{friends_KLdiv}
\end{figure}
The KL divergence is generally low, indicating good approximations all around. There are some instances where the SVI algorithms does better than methods relying on the adjusted pseudolikelihood, such as $\theta_3$ of all three models. For SVI (a), taking $K \geq 5$ seems to give more stable results while $K\geq100$ seems to be sufficient for SVI (b).

Plots of the marginal posterior distributions in Figure \ref{friends_posterior}
\begin{figure}[htb!]
\centering
\includegraphics[width=0.9\textwidth]{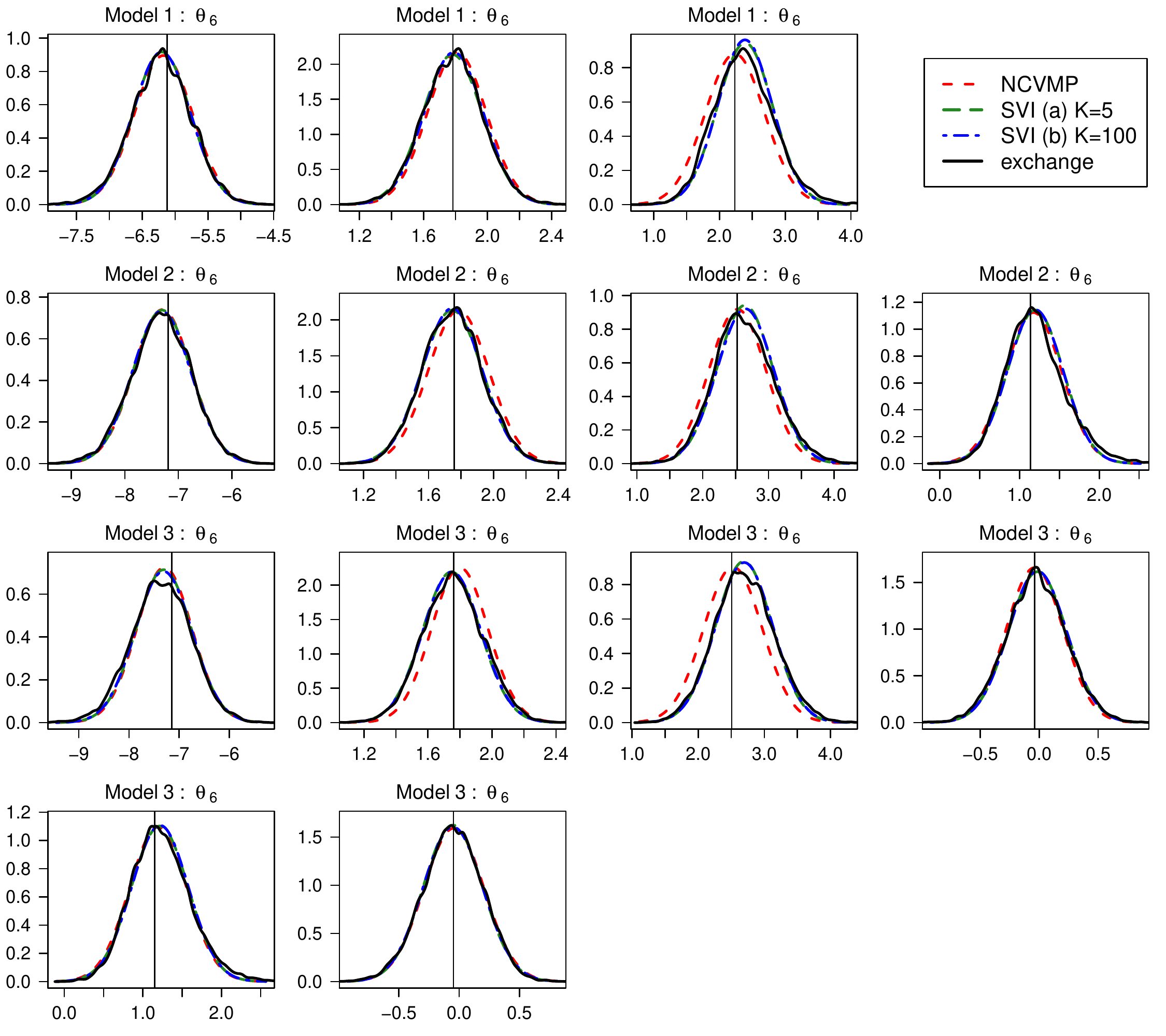}
\caption{\small Friends network. Marginal posterior distributions of each $\theta_j$ from variational and exchange algorithms. Vertical lines indicate values of $\hat{\theta}_\ML$.} \label{friends_posterior}
\end{figure}
confirm the above observations. The approximate marginal posteriors are generally very close to that estimated by the exchange algorithm. There is slight overestimation of the posterior mean of $\theta_2$ and underestimation of the posterior mean of $\theta_3$ in all three models by NCVMP. The tendency of NCVMP to lock on to $\hat{\theta}_\ML$ is still observed.

Estimates of the IWLBs and log marginal likelihood by CJ's method are shown in Table \ref{friends_IWLB}. 
\begin{table}[htb!]
\centering
\begin{small}
\begin{tabular}{ll|ccc|ccc}
\hline
Method & $K$ & \multicolumn{3}{c}{Computation time (seconds)}  & \multicolumn{3}{c}{V} \\
&& $\mathcal{M}_1$ &  $\mathcal{M}_2$ & $\mathcal{M}_3$ & $\mathcal{M}_1$ &  $\mathcal{M}_2$ & $\mathcal{M}_3$\\
\hline
CJ && -235.5 & -231.8 & -239.5 \\
Laplace & & -235.5 (1.0) & -231.8 (1.1) & -239.5 (1.7) & 100 & 100 & 100 \\
NCVMP & & -235.5 (1.0) & -231.8 (1.2) & -239.5 (1.8) & 100 & 100 & 100 \\
\multirow{3}{*}{SVI (a)} & 1 & -235.5 (3.7) & -231.6 (3.7) & -239.0 (5.9) & (100, 150) & (100, 150) & (200, 250) \\
& 5 & -235.5 (3.7) & -231.6 (3.5) & -239.0 (6.4) & (100, 150) & (100, 150) & (200, 350) \\
& 20 & -235.5 (3.1) & -231.6 (3.7) & -239.0 (5.7) & (100, 150) & (100, 150) & (200, 250) \\
\multirow{3}{*}{SVI (b)} & 100 & -235.5 (3.4) & -231.6 (3.3) & -239.0 (5.9) & (100, 150) & 150 & (200, 250) \\
& 200 & -235.5 (3.6) & -231.6 (3.5) & -239.0 (5.7) & (100, 150) & 150 & (200, 250) \\
& 500 & -235.5 (3.4) & -231.6 (3.5) & -239.0 (5.9) & (100, 150) & (100, 150) & (200, 250) \\
\hline
\end{tabular}
\end{small}
\caption{\small Friends network. IWLB computed using (I) adjusted pseudolikelihood for Laplace and NCVMP and (II) Monte Carlo for SVI algorithms. Computation times (in brackets) and range of values of $V$.} \label{friends_IWLB}
\end{table}
The IWLBs estimated using Laplace approximation and NCVMP based on approach (I) totally agree with CJ's method. However, there are some minor discrepancies in the IWLBs estimated using the SVI algorithms based on approach (II) with CJ's method; the IWLBs are slightly higher for $\M_2$ and $\M_3$. Using $\M_1$ as reference, the Bayes factor $B_{21}$ ranges between 40.4--49.4 while $B_{31}$ ranges between 0.02--0.03. Hence, the preferred model is $\M_2$ and we conclude that the observed network can be explained by the homophily effect of drug usage but not that of sports and smoking.

\subsection{E. coli network}
Here we consider the E.coli transcriptional regulation network \citep{Shen2002} based on the RegulonDB data \citep{Salgado2001}. This biological network is available as {\tt data(ecoli)} from the {\tt ergm} R package, and has been analyzed using ERGMs by \cite{Saul2007} and \cite{Hummel2012} among others. The undirected version of the network has 419 nodes representing operons and 519 edges representing regulating relationships. 

For estimating the parameters in the adjusted pseudolikelihood, we use $10^5$ auxiliary iterations and a thinning factor of 1000 to simulate 1000 samples from $p(y|\hat{\theta}_\ML)$. This computation took 4.5, 3.2 and 5.6 seconds for $\M_1$, $\M_2$ and $\M_3$ respectively. Next, we fit the three models in \eqref{models} using the exchange algorithm, CJ's method, Laplace approximation and the variational algorithms. For the exchange algorithm, we also use $10^5$ auxiliary iterations. The length of burn-in is set as 1000 and the number of iterations per chain excluding burn-in is 10000. Thus we have 40000 samples from 4 chains for $\M_1$ and $\M_2$, and 60000 samples from 6 chains for $\M_3$. For this network, we were unable to tune $\gamma$ so that the acceptance rates fall between 20--25\% using the default value of 0.0025 for $\sigma_\epsilon$ in the {\tt bergm} function. After repeated tries, the average acceptance rates are 22.1\%, 22.9\%, 23.4\% if we set $\sigma_\epsilon$ as 0.002, 0.002, 0.0015 and $\gamma$ as 1.0, 0.1, 0.1 for $\M_1$, $\M_2$ and $\M_3$ respectively. For CJ's method, the acceptance rates for $\M_1$, $\M_2$ and $\M_3$ after tuning are 23.0\%, 22.9\% and 24.1\% respectively. 

Computation times for the various algorithms are shown in Table \ref{ecoli_computation_times}. As before, methods relying on the adjusted pseudolikelihood are the fastest followed by SVI (b), SVI (a) and lastly the exchange algorithm. SVI (b) converged very fast even for this large network, as the number of particles in $\mathbb{S}$ is quite small and simulation from the likelihood is significantly reduced compared to SVI(a). The number of particles in $\mathbb{S}$ for $\M_3$ is about 3--4 times the number for $\M_1$ and $\M_2$. Hence, the speedup of SVI (b) as compared to SVI (a) is typically lower when $\theta$ is higher in dimension. 
\begin{table}[tb!]
\centering
\begin{small}
\begin{tabular}{ll|cccl}
\hline
Method & $K$ & \multicolumn{3}{c}{Computation time (seconds)}  \\
&& $\mathcal{M}_1$ &  $\mathcal{M}_2$ & $\mathcal{M}_3$ \\ \hline
Exchange && 2719.8 & 809.7 & 8827.5 \\ \hline
CJ && 3.5 & 5.8 & 10.1 \\
Laplace && 0.0 & 0.0 & 0.0 \\
NCVMP && 0.0 & 0.2 & 0.4 \\ \hline
\multirow{3}{*}{SVI (a)} & 1&  83.0  $\pm$  0.4 & 80.9  $\pm$  26.3 & 116.4  $\pm$  24.4 \\
& 5 & 85.3  $\pm$  0.4 & 76.5  $\pm$  16.6 & 104.1  $\pm$  0.1 \\
& 20 & 93.7  $\pm$  0.3 & 79.4  $\pm$  23.1 & 120.6  $\pm$  18.1 \\ \hline
\multirow{3}{*}{SVI (b)} & 100 &  2.2 $\pm$ 0.2 (18, 24) & 2.2 $\pm$ 0.2 (27, 34) & 8.3 $\pm$ 0.6 (77, 92) \\
& 200 & 3.1 $\pm$ 0.3 (19, 28) & 2.9 $\pm$ 0.3 (28, 38) & 13.2 $\pm$ 1.4 (81, 110) \\
& 500 & 5.2 $\pm$ 0.5 (20, 29) & 5.1 $\pm$ 0.4 (29, 40) & 27.8 $\pm$ 2.8 (97, 128) \\ \hline
\end{tabular}
\end{small}
\caption{\small E. coli network. Computation times of various algorithms. For SVI (b), the range of the number of particles in $\mathbb{S}$ over ten runs is given in brackets.} \label{ecoli_computation_times}
\end{table}

Figure \ref{ecoli_KLdiv} shows how close the approximate marginal posteriors of each $\theta_i$ are in KL divergence to that estimated using the exchange algorithm. The methods relying on adjusted pseudolikelihood (NCVMP, Laplace and CJ) have similar performance and there is often no clear winner, although NCVMP performs consistently better than Laplace and CJ for $\M_3$. The SVI algorithms perform slightly better than NCVMP, Laplace and CJ for $\theta_2$ of $\M_1$ and $\{\theta_1, \theta_3\}$ of $\M_3$, and significantly better for $\M_2$. There are a couple of outliers from SVI (b) for $K=100$ and $K \geq 200$ seems to be more stable.
\begin{figure}[tb!]
\centering
\includegraphics[width=\textwidth]{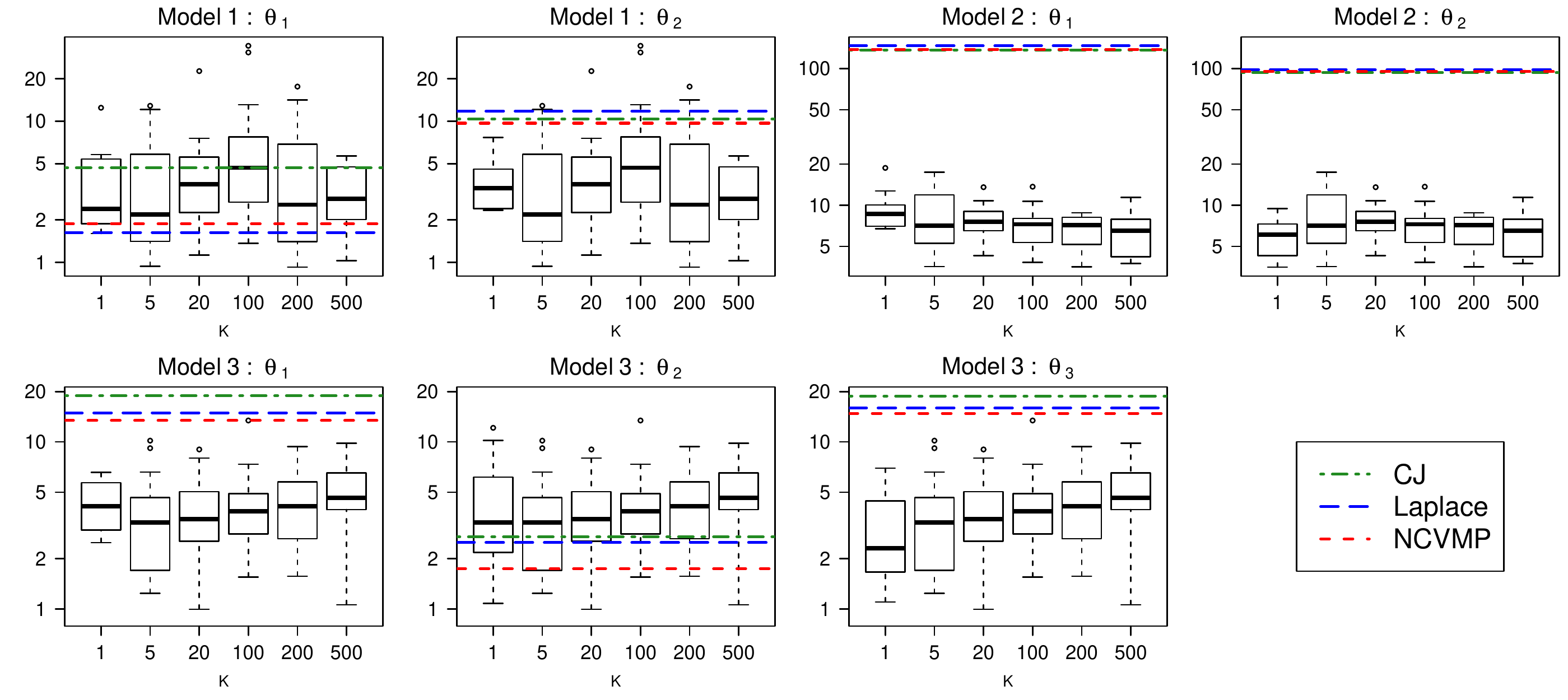}
\caption{\small E. coli network. KL divergence of approximate marginal posterior of each $\theta_j$ from posterior estimated using exchange algorithm. In each figure, the first three and last three boxplots correspond to SVI (a) and SVI (b) respectively.} \label{ecoli_KLdiv}
\end{figure}

Figure \ref{ecoli_posterior} shows the marginal posterior distributions from NCVMP and one instance of SVI (a) ($K=5$) and SVI (b) ($K=200$). There is slight overestimation of the posterior variance for $\theta_2$ of $\M_1$ and $\{\theta_1, \theta_3\}$ of $\M_3$. For $\M_2$, the posterior mean and variance of $\theta_1, \theta_2$ are not captured accurately. NCVMP appears to lock on to $\hat{\theta}_\ML$ too tightly again.
\begin{figure}[htb!]
\centering
\includegraphics[width=0.9\textwidth]{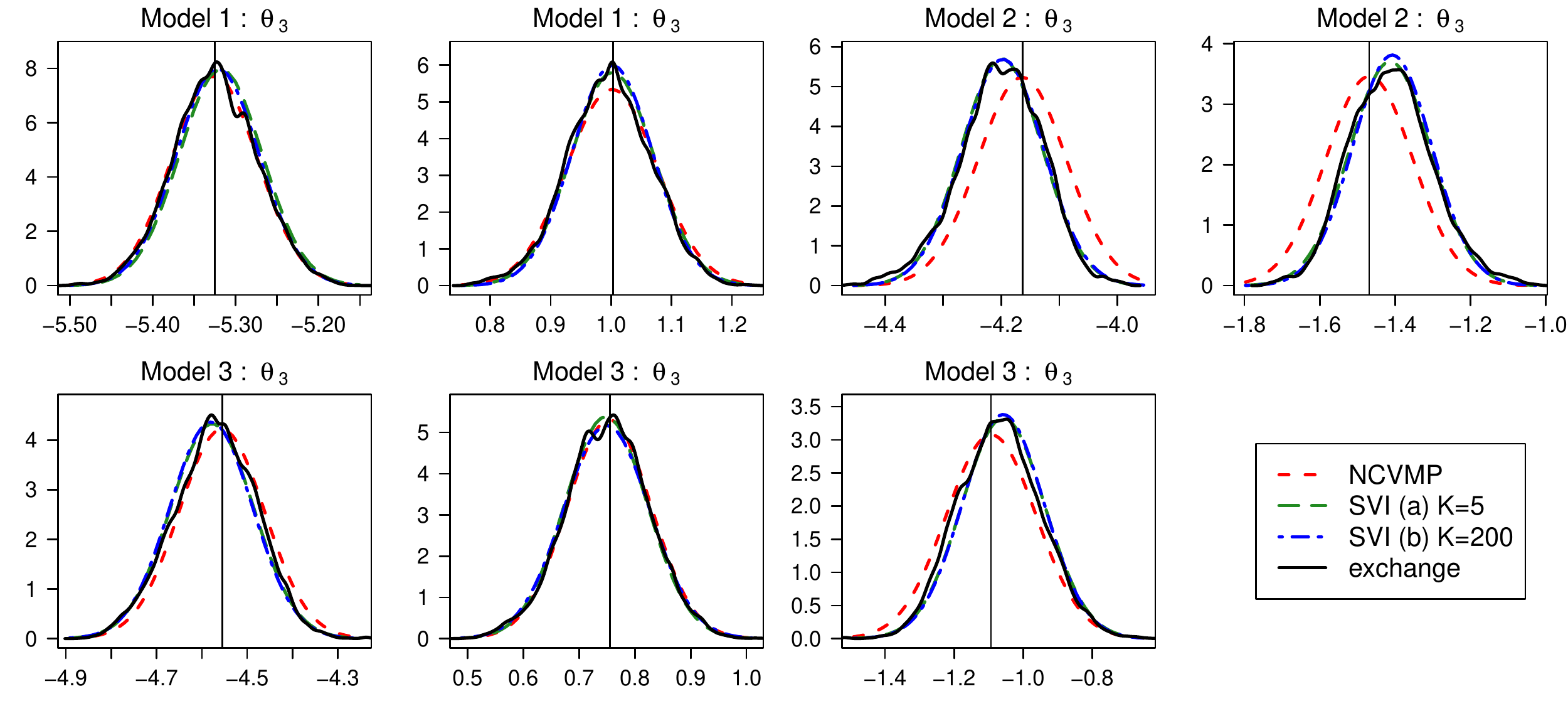}
\caption{\small E. coli network. Marginal posterior distributions of each $\theta_j$ from variational and exchange algorithms. Vertical lines indicate values of $\hat{\theta}_\ML$.} \label{ecoli_posterior}
\end{figure}
\begin{figure}[tb!]
\centering
\includegraphics[width=\textwidth]{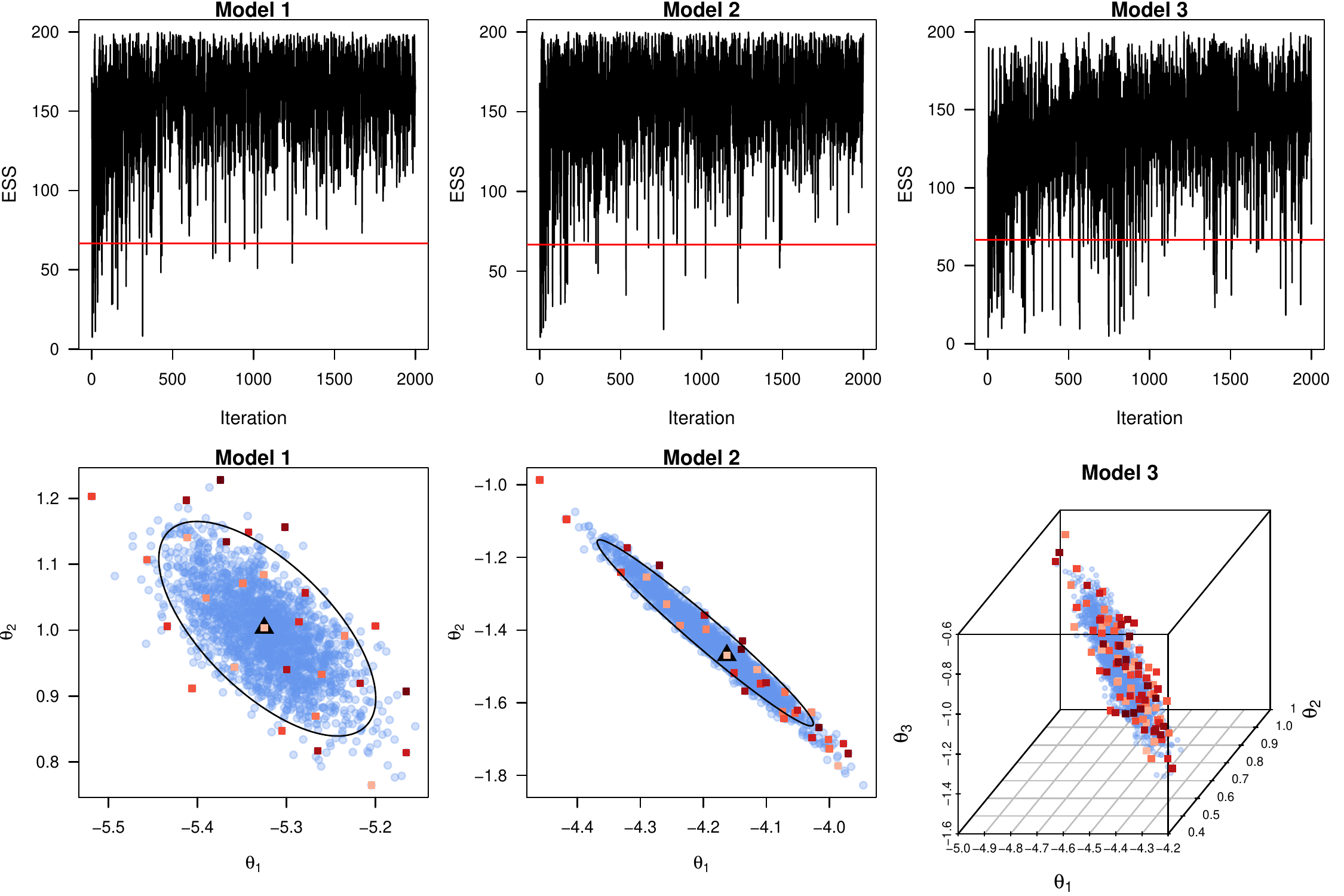}
\caption{\small E. coli network. First row shows the ESS at each iteration for a run of SVI (b) ($K=200)$. Second row show plots of $\{\theta^{(t)}\}$ (in blue circles) and particles in $\mathbb{S}$ (red squares). Maximum likelihood estimate is marked using a black triangle. Ellipse in black is the normal density contour that contain 95\% of the probability of the estimated $q_\lambda(\theta)$.} \label{ecoli_ESS_particles}
\end{figure}
Figure \ref{ecoli_ESS_particles} illustrates the performance of one run of SVI (b) ($K=200)$ for each of the three models. The number of particles in $\mathbb{S}$ are 28, 31 and 97 respectively. For $\M_1$, most of the particles in $\mathbb{S}$ are added within the first 500 iterations and hence convergence will be fairly rapid after that. For the higher dimensional $\M_3$, $\mathbb{S}$ was still expanding after the first 1000 iterations. The second row shows the spread of the particles in $\mathbb{S}$. Particles colored in lighter shades of red are added to $\mathbb{S}$ earlier. Particles which are added later (darker shade of red) have a higher tendency to appear on the boundary for $\M_1$ and $\M_2$.

Table \ref{ecoli_IWLB} shows the IWLBs and log marginal likelihood estimated using CJ's method. The IWLB from Laplace and NCVMP estimated using approach (I) are identical to CJ's method. For the SVI algorithms, the IWLB estimated using approach (II) are also identical to CJ's method for $\M_1$ and $\M_2$, but is slightly higher for $\M_3$. Only a small $V$ of up to 100 was required for the IWLB algorithm to converge. The preferred model in this case is $\M_3$, which contains both the gwesp and gwd terms.
\begin{table}[htb!]
\centering
\begin{small}
\begin{tabular}{ll|ccc|ccc}
\hline
Method & $K$ & \multicolumn{3}{c}{Computation time (seconds)}  & \multicolumn{3}{c}{V} \\
&& $\mathcal{M}_1$ &  $\mathcal{M}_2$ & $\mathcal{M}_3$ & $\mathcal{M}_1$ &  $\mathcal{M}_2$ & $\mathcal{M}_3$\\
\hline
CJ && -3123.8 & -3130.6 & -3097.5 \\
Laplace & & -3123.8 (0.2) & -3130.6 (0.7) & -3097.5 (3.2) & 50 & 50 & 100 \\
NCVMP & & -3123.8 (0.1) & -3130.6 (0.8) & -3097.5 (1.6) & 50 & 50 & 50 \\ \hline
\multirow{3}{*}{SVI (a)}  & 1 & -3123.8 (1.8) & -3130.6 (2.7) & -3097.2 (2.7) & (50, 100) & 100 & 100 \\
& 5 & -3123.8 (1.6) & -3130.6 (2.7) & -3097.2 (2.7) & (50, 100) & 100 & 100 \\
& 20 & -3123.8 (1.6) & -3130.6 (2.7) & -3097.2 (2.7) & (50, 100) & 100 & 100 \\ \hline
\multirow{3}{*}{SVI (b)}  & 100 & -3123.8 (1.9) & -3130.6 (2.7) & -3097.2 (2.7) & (50, 100) & 100 & 100 \\ 
& 200 & -3123.8 (1.8) & -3130.6 (2.7) & -3097.2 (2.7) & (50, 100) & 100 & 100 \\
& 500 & -3123.8 (1.3) & -3130.6 (2.7) & -3097.2 (2.7) & 50 & 100 & 100 \\
\hline
\end{tabular}
\end{small}
\caption{\small E. coli network. IWLB computed using (I) adjusted pseudolikelihood for Laplace and NCVMP and (II) Monte Carlo for SVI algorithms. Computation times (in brackets) and range of values of $V$.} \label{ecoli_IWLB}
\end{table}

 \section{Conclusion} \label{sec:Conclu}
In this article, we have proposed several variational methods for obtaining Bayesian inference for the ERGM. The first approach is an NCVMP algorithm which approximates the likelihood using an adjusted pseudolikelihood. NCVMP is extremely fast and stable as it considers deterministic updates. Comparing NCVMP with Laplace approximation, which is also deterministic and yields a Gaussian approximation, the performance of the two approaches are quite similar in many cases. Sometimes NCVMP is more accurate than Laplace and sometimes it is the other way around, but both approaches are dependent on the adjusted pseudolikelihood and can only provide good posterior approximations when the adjusted pseudolikelihood is able to mimic the true likelihood well. NCVMP has a tendency to lock on to $\hat{\theta}_\ML$ too tightly even when the true posterior mean deviates from $\hat{\theta}_{\ML}$, although the approximation is still very close to results from the exchange algorithm in many cases. In the second approach, we develop a SVI algorithm. As simulating from the likelihood is very time consuming, we estimate the gradients in two ways (a) using a Monte Carlo estimate based on a small number of samples and (b) adaptive SNIS. The SVI algorithms are very fast and yield posterior approximations which are very close to results from the exchange algorithm in all our experiments. For SVI (a), only a small number of simulations ($K \approx 5$) are required at each iteration. SVI (b) is much faster than SVI (a) for low-dimensional problems ($p \leq 4$) but the computational advantage becomes smaller as the dimension of $\theta$ increases. A collection of particles and associated sufficient statistics has to be stored for SVI (b), but only a small number of simulations (100 or 200) per particle is required for networks considered in this article. Using the variational or Laplace approximation, we can also compute an importance weighted lower bound, which is identical to the log marginal likelihood estimate from CJ's method in our experiments when (I) the adjusted pseudolikelihood is used. This can be useful in model selection when a large number of candidate models are compared, as unlike sampling-based methods, tuning of parameters or checking of diagnostic plots for convergence is not required. The IWLB based on (II) Monte Carlo estimate of log normalizing constant can also be useful when the adjusted pseudolikelihood is unable to mimic the true likelihood well. There remains many avenues open for exploration, such as the use of multiple or mixture importance sampling in computing gradient estimates in SVI (b) and improved criteria for assessing the performance of importance sampling estimates.

\section*{Acknowledgments}
Linda Tan was supported by the start-up grant R-155-000-190-133. The Insight Centre for Data Analytics is supported by Science Foundation Ireland under Grant Number SFI/12/RC/2289. We thank the editor, associate editor and reviewers for their comments which have greatly improved this manuscript.

\setcounter{section}{0} \renewcommand{\thesection}{S\arabic{section}}
\setcounter{figure}{0} \renewcommand{\thefigure}{S\arabic{figure}}
\setcounter{table}{0} \renewcommand{\thetable}{S\arabic{table}}
\setcounter{equation}{0} \renewcommand{\theequation}{S\arabic{equation}}

\newpage
\onehalfspacing
\begin{center}
{\Large\bf Supplementary material for ``Bayesian variational  \\
inference for exponential random graph models"}
\end{center}

\section{Gradient and Hessian of the ERGM log likelihood}
The log likelihood of the ERGM is
\begin{equation*}
\log p(y|\theta) = \theta^T s(y) - \log z(\theta).
\end{equation*}
The gradient is given by 
\begin{equation*}
\begin{aligned}
 \nabla_\theta \log p(y|\theta) &= s(y) - \nabla_\theta z(\theta) /z(\theta) \\
& = s(y) - E_{y|\theta} [s(y)],
\end{aligned}
\end{equation*}
since
\begin{equation*} 
\frac{\nabla_\theta z(\theta)}{z(\theta)} 
= \frac{ \sum_{y \in \mathcal{Y}} \exp\{\theta^T s(y)\} s(y)}{z(\theta)} 
= \sum_{y \in \mathcal{Y}} p(y|\theta) s(y) = E_{y|\theta} [s(y)].
\end{equation*}
The Hessian is given by 
\begin{equation*}
\begin{aligned}
 \nabla_\theta^2 \log p(y|\theta) &= - \bigg\{ \frac{\nabla_\theta^2 z(\theta)}{z(\theta)} - \frac{\nabla_\theta z(\theta)}{z(\theta)} \frac{\nabla_\theta z(\theta)^T}{z(\theta)} \bigg\}\\
& = - \{E_{y|\theta} [s(y)s(y)^T] - E_{y|\theta} [s(y)]E_{y|\theta} [s(y)]^T\} \\
&= - \cov_{y|\theta} [s(y)].
\end{aligned}
\end{equation*}
since
\begin{equation*} 
\frac{\nabla_\theta^2 z(\theta)}{z(\theta)} 
= \frac{ \sum_{y \in \mathcal{Y}} \exp\{\theta^T s(y)\} s(y)s(y)^T}{z(\theta)} 
= \sum_{y \in \mathcal{Y}} p(y|\theta) s(y)s(y)^T = E_{y|\theta} [s(y)s(y)^T].
\end{equation*}

\section{Nonconjugate variational message passing}
Nonconjugate variational message passing is a fixed-point iteration method, which is sensitive to initialization and is not guaranteed to converge. However, it can also be interpreted as a natural gradient method with a step size of one and smaller steps can hence be used to alleviate convergence issues \citep{Tan2014}. 

\subsection{Natural gradient method}
The natural gradient of the lower bound $\L$ with respect to $\lambda$ can be obtained by pre-multiplying the (Euclidean) gradient by the inverse of the Fisher information matrix of $q_\lambda(\theta)$ \citep{Amari1998}, which is given by $\V(\lambda)$. Hence, from \eqref{euclidean_gradient}, the natural gradient is
\begin{equation*}
\begin{aligned}
\widetilde{\nabla}_\lambda \L &= \V(\lambda)^{-1} \nabla_\lambda \L  \\
&= \V(\lambda)^{-1}  \nabla_\lambda E_{q_\lambda} \{ \log p(\theta,y) \} - \lambda \\
&= \hat{\lambda} - \lambda,
\end{aligned}
\end{equation*}
where $\hat{\lambda}$ is the update used in nonconjugate variational message passing. If we replace the {Euclidean} gradient in gradient ascent by the natural gradient, then at the $t$th iteration, 
\begin{equation} \label{natural_gradient_step}
\begin{aligned}
\lambda^{(t+1)} &= \lambda^{(t)} + \rho_t (\hat{\lambda}^{(t)} -  \lambda^{(t)})  \\
&= (1-\rho_t) \lambda^{(t)} + \rho_t \hat{\lambda}^{(t)}.
\end{aligned}
\end{equation}
Thus nonconjugate variational message passing is just the special case in natural gradient ascent where the stepsize $\rho_t = 1$. 

For ERGMs, $\hat{\lambda}^{(t)}$ cannot be evaluated in closed form and we approximate it deterministically using Gauss-Hermite quadrature. This is unlike stochastic gradient ascent, where the gradient is a noisy but {\em unbiased estimate} of the true gradient, and the updates are guaranteed to converge to the true optimum provided the objective function and the stepsize satisfy certain regularity conditions. While Algorithm 1 does not enjoy such guarantees,  we are able to compute an approximation of the lower bound at each iteration of Algorithm 1, and use this as a means to assess whether the algorithm is converging towards a local maximum. If the lower bound fails to increase, we can attempt to resolve this issue by repeatedly halving the stepsize $\rho_t$ in \eqref{natural_gradient_step} or try some other initialization.

\begin{theorem}
The natural gradient ascent update for a multivariate Gaussian $q_\lambda(\theta)  = N(\mu, \Sigma)$. can be expressed as 
\begin{equation*}
\begin{aligned}
{\Sigma^{(t+1)}}^{-1} &= (1-\rho_t){\Sigma^{(t)}}^{-1} + \rho_t \vec^{-1} \bigg[- 2\frac{\partial \H}{\partial  \vec(\Sigma)}\bigg] \\
\mu^{(t+1)} &= \mu^{(t)} + \rho_t \Sigma^{(t+1)}\frac{\partial \H}{\partial  \mu},
\end{aligned}
\end{equation*}
where $\H = \log p(y, \theta)$. To apply damping, $0 < \rho_t < 1$ can be used.
\end{theorem}
\begin{proof}
From \cite{Tan2013}, we can write $q_\lambda(\theta)$ in the form of a exponential family distribution as
\begin{equation*}
q_\lambda(\theta_i)=\exp\{\lambda^T t(\theta)-h(\lambda)\},
\end{equation*}
where 
\begin{equation*}
\lambda = \begin{bmatrix}  -\frac{1}{2}D_p^T \vec({\Sigma}^{-1}) \\ {\Sigma}^{-1}\mu \end{bmatrix}, \quad 
t(\theta)  = \begin{bmatrix} \vech(\theta\theta^T)\\
\theta \end{bmatrix} 
\end{equation*}
and $h(\lambda)= \frac{1}{2} \mu^T\Sigma^{-1}\mu +\frac{1}{2}\log|\Sigma|+\frac{p}{2}\log(2\pi)$. The $p^2 \times p(p+1)/2$ duplication matrix $D_p$ is such that  $D_p \vech(A)=\text{vec}(A)$ if $A$ is symmetric. Let $D_p^+$ denote the Moore-Penrose inverse of $D_p$. \cite{Tan2013} showed that the NCVMP update is 
\begin{equation*}
\hat{\lambda} = \begin{bmatrix}
D_p^T  & 0 \\
-2(\mu^T \otimes I){D_p^+}^T \negthinspace D_p^T & I 
\end{bmatrix}
\begin{bmatrix}
\frac{\partial \H}{\partial  \text{vec}(\Sigma)} \\ 
\frac{\partial \H}{\partial  \mu} \end{bmatrix}.
\end{equation*}  
Now $\lambda^{(t+1)} = (1-\rho_t) \lambda^{(t)} + \rho_t \hat{\lambda}^{(t)}$ implies
\begin{equation*}
\begin{bmatrix}  -\frac{1}{2}D_p^T \text{vec}({\Sigma^{(t+1)}}^{-1}) \\ {\Sigma^{(t+1)}}^{-1}\mu^{(t+1)} \end{bmatrix} = (1-\rho_t) \begin{bmatrix}  -\frac{1}{2}D_p^T \text{vec}({\Sigma^{(t)}}^{-1}) \\ {\Sigma^{(t)}}^{-1}\mu^{(t)} \end{bmatrix} + \rho_t \begin{bmatrix}
D_p^T \frac{\partial \H}{\partial  \text{vec}(\Sigma)} \\ -2({\mu^{(t)}}^T \otimes I){D_p^+}^T \negthinspace D_p^T \frac{\partial \H}{\partial  \text{vec}(\Sigma)} + \frac{\partial \H}{\partial  \mu}
\end{bmatrix}.
\end{equation*}
Note that $D_p^+ = (D_p D_p)^{-1} D_p^T$ and $D_pD_p^+ = (I_{p^2} + K_p)/2$, where $K_p$ is the commutation matrix such that $K_p \vec(A) = \vec(A^T)$. See \cite{Magnus1999}. Premultiplying the first line by $-2 D_p(D_p^TD_p)^{-1}$, we have 
\begin{equation*}
\begin{gathered}
D_pD_p^+ \text{vec}({\Sigma^{(t+1)}}^{-1}) = (1-\rho_t) D_pD_p^+ \text{vec}({\Sigma^{(t)}}^{-1}) -  2\rho_t D_pD_p^+ \frac{\partial \H}{\partial  \text{vec}(\Sigma)}, \\
\implies \text{vec}({\Sigma^{(t+1)}}^{-1}) = (1-\rho_t) \text{vec}({\Sigma^{(t)}}^{-1}) - 2\rho_t  \frac{\partial \H}{\partial  \text{vec}(\Sigma)}, \\
\implies {\Sigma^{(t+1)}}^{-1} = (1-\rho_t){\Sigma^{(t)}}^{-1} + \rho_t \vec^{-1} \bigg[- 2\frac{\partial \H}{\partial  \text{vec}(\Sigma)}\bigg].
\end{gathered}
\end{equation*}
For the second line,
\begin{equation*}
\begin{aligned}
&{\Sigma^{(t+1)}}^{-1}\mu^{(t+1)} = (1-\rho_t) {\Sigma^{(t)}}^{-1}\mu^{(t)} + \rho_t \bigg[\frac{\partial \H}{\partial  \mu} -2({\mu^{(t)}}^T \otimes I) \frac{\partial \H}{\partial  \text{vec}(\Sigma)} \bigg], \\
&\implies {\Sigma^{(t+1)}}^{-1}\mu^{(t+1)} = (1-\rho_t) {\Sigma^{(t)}}^{-1}\mu^{(t)} + \rho_t \frac{\partial \H}{\partial  \mu} \\
&\hspace{30mm} \qquad \qquad +  ({\mu^{(t)}}^T \otimes I)\bigg[ \text{vec}({\Sigma^{(t+1)}}^{-1}) - (1-\rho_t) \text{vec}({\Sigma^{(t)}}^{-1}) \bigg], \\
&\implies {\Sigma^{(t+1)}}^{-1}\mu^{(t+1)} = (1-\rho_t) {\Sigma^{(t)}}^{-1}\mu^{(t)} + \rho_t \frac{\partial \H}{\partial  \mu} + {\Sigma^{(t+1)}}^{-1} \mu^{(t)} - (1-\rho_t) {\Sigma^{(t)}}^{-1} \mu^{(t)}, \\
& \implies\mu^{(t+1)}= (1-\rho_t) \Sigma^{(t+1)} {\Sigma^{(t)}}^{-1}\mu^{(t)} + \rho_t \Sigma^{(t+1)}\frac{\partial \H}{\partial  \mu} +\mu^{(t)} - (1-\rho_t) \Sigma^{(t+1)}{\Sigma^{(t)}}^{-1} \mu^{(t)}, \\
&\implies\mu^{(t+1)} = \mu^{(t)} + \rho_t \Sigma^{(t+1)}\frac{\partial \H}{\partial  \mu}.
\end{aligned}
\end{equation*}
\end{proof}

\subsection{Lower bound using pseudolikelihood as plug-in}

Using the adjusted pseudolikelihood $\tilde{f}(y|\theta)$ as a plug-in for the true likelihood $p(y|\theta)$, 
\begin{equation} \label{logjoint}
\begin{aligned}
\log \tilde{p}(y, \theta) &= \log M + \sum_{(i,j) \in \D} \{ y_{ij} \delta_s(y)_{ij}^Tg(\theta)  -b[ \delta_s(y)_{ij}^Tg(\theta) ] \} \\
&\quad -  \tfrac{p}{2}\log (2\pi) -  \tfrac{1}{2}\log|\Sigma_0| - \tfrac{1}{2}(\theta -  \mu_0)^T \Sigma_0^{-1}( \theta - \mu_0) ,
\end{aligned}
\end{equation}
where $g(\theta) = \hat{\theta}_{\PL} + W(\theta - \hat{\theta}_{\ML})$ and $b(x) = \log\{1+\exp(x)\}$. Let 
\begin{equation*}
\begin{aligned}
\delta_s(y)_{ij}^Tg(\theta) &= \delta_s(y)_{ij}^T \{( \hat{\theta}_{\PL} - W \hat{\theta}_{\ML}) + W\theta \} \\
&= \alpha_{ij} + \beta_{ij}^T \theta,
\end{aligned}
\end{equation*}
where $\alpha_{ij} = \delta_s(y)_{ij}^T ( \hat{\theta}_{\PL}  - W\hat{\theta}_{\ML})$ and $\beta_{ij} = W^T \delta_s(y)_{ij}$. Then
\begin{equation*}
\begin{aligned}
E_{q_\lambda} \{\log \tilde{p}(y, \theta) \} &= \log M + \sum_{(i,j) \in \D} [y_{ij} (\alpha_{ij} + \beta_{ij}^T \mu)  - E_{q_\lambda} \{b(\alpha_{ij} + \beta_{ij}^T \theta) \}] \\
& \quad -  \tfrac{p}{2}\log (2\pi) - \tfrac{1}{2}\log|\Sigma_0| -  \tfrac{1}{2}(\mu - \mu_0)^T \Sigma_0^{-1}(\mu - \mu_0) -  \tfrac{1}{2}\tr(\Sigma_0^{-1} \Sigma).
\end{aligned}
\end{equation*}
As 
\begin{equation*}
E_{q_\lambda} \{ \log q_\lambda (\theta) \} = - \tfrac{p}{2}\log (2\pi) - \tfrac{1}{2}\log|\Sigma| - \tfrac{p}{2},  
\end{equation*}
the approximate lower bound is given by
\begin{equation*}
\begin{aligned}
\tilde{\L} & = E_{q_\lambda} \{\log \tilde{p}(y, \theta) - \log q_\lambda (\theta) \}\\
&= \log M +  \sum_{(i,j) \in \D} [ y_{ij} (\alpha_{ij} + \beta_{ij}^T\mu )- E_{q_\lambda}\{ b(\alpha_{ij} +\beta_{ij}^T\theta) \} ] - \frac{1}{2} \log|\Sigma_0|\\
&\quad - \tfrac{1}{2}(\mu - \mu_0)^T \Sigma_0^{-1}(\mu- \mu_0) - \tfrac{1}{2} \tr(\Sigma_0^{-1} \Sigma) +\tfrac{1}{2} \log|\Sigma| + \tfrac{p}{2}.
\end{aligned}
\end{equation*}

\subsection{Gradients in NCVMP}
We can write
\begin{equation*}
\begin{aligned}
E_{q_\lambda} \{\log \tilde{p}(y, \theta) \} &= \log M + \sum_{(i,j) \in \D} [y_{ij} m_{ij} - B^{(0)}(m_{ij}, v_{ij})] -  \tfrac{p}{2}\log (2\pi)\\
& \quad  - \tfrac{1}{2}\log|\Sigma_0| -  \tfrac{1}{2}(\mu - \mu_0)^T \Sigma_0^{-1}(\mu - \mu_0) -  \tfrac{1}{2}\tr(\Sigma_0^{-1} \Sigma).
\end{aligned}
\end{equation*}
Let $d$ denote the differential operator. Differentiating w.r.t. $\mu$, 
\begin{equation*}
\begin{gathered}
d E_{q_\lambda} \{\log \tilde{p}(y, \theta) \} =  \sum_{(i,j) \in \D} [y_{ij} - \int_{-\infty}^\infty b^{(1)} (v_{ij}z + m_{ij}) \phi(z|0,1) dz] \beta_{ij}^T d\,\mu - (\mu - \mu_0)^T \Sigma_0^{-1}d\mu. \\
\implies \nabla_\mu E_{q_\lambda} \{\log \tilde{p}(y, \theta) \} = \sum_{(i,j) \in \D} [y_{ij}  - B^{(1)} (m_{ij}, v_{ij}) ] \beta_{ij} -\Sigma_0^{-1} (\mu - \mu_0).
\end{gathered}
\end{equation*}
Differentiating w.r.t. $\vec(\Sigma)$,
\begin{equation*}
\begin{aligned}
d \,E_{q_\lambda} \{\log \tilde{p}(y, \theta) \} &= - \sum_{(i,j) \in \D} \int_{-\infty}^\infty b^{(1)} (v_{ij}z + m_{ij}) \phi(z|0,1) z dz (d v_{ij})-  \tfrac{1}{2} \vec(\Sigma_0^{-1})^T d \vec(\Sigma).
\end{aligned}
\end{equation*} 
We have $2v_{ij} dv_{ij} = \vec(\beta_{ij} \beta_{ij})^T d\vec(\Sigma)$. Using integration by parts,
\begin{equation*}
\begin{aligned}
\int_{-\infty}^\infty b^{(1)}(v_{ij} z + m_{ij}) z\, \phi(z|0,1) dz 
&=  v_{ij}\int_{-\infty}^\infty b^{(2)}(v_{ij} z + m_{ij}) \phi(z|0,1) dz \\
&= v_{ij} B^{(2)}(m_{ij}, v_{ij}).
\end{aligned}
\end{equation*}
Hence
\begin{equation*}
\begin{gathered}
d \,E_{q_\lambda} \{\log \tilde{p}(y, \theta) \} = - \tfrac{1}{2} \sum_{(i,j) \in \D} B^{(2)}(m_{ij}, v_{ij}) \vec(\beta_{ij} \beta_{ij})^T d\vec(\Sigma)-  \tfrac{1}{2} \vec(\Sigma_0^{-1})^T d \vec(\Sigma) \\
\implies \nabla_{\vec(\Sigma)} E_{q_\lambda} \{\log \tilde{p}(y, \theta) \} =  - \tfrac{1}{2} \sum_{(i,j) \in \D} B^{(2)}(m_{ij}, v_{ij}) \vec(\beta_{ij} \beta_{ij}) -  \tfrac{1}{2} \vec(\Sigma_0^{-1}).
\end{gathered}
\end{equation*}

\section{Gradients of variational density}
We have 
\begin{equation*}
\log q_\lambda (\theta) = -\tfrac{p}{2} \log(2\pi) - \log|C| - \tfrac{1}{2} (\theta-\mu)^T C^{-T} C^{-1} (\theta - \mu).
\end{equation*}
Differentiating w.r.t. $\theta$,
\begin{equation*}
d\, \log q_\lambda (\theta)  = - (\theta-\mu)^T C^{-T} C^{-1} d\, \theta = - s^T C^{-1} d\,\theta.
\end{equation*}
Hence $\nabla_\theta \log q_\lambda (\theta) = C^{-T}s$. Similarly, $\nabla_\mu \log q_\lambda (\theta) = -C^{-T}s$. Differentiating w.r.t. $C$,
\begin{equation*}
\begin{aligned}
d\, \log q_\lambda (\theta)  &= - \tr(C^{-1}d\, C) + \tfrac{1}{2}s^T(C^{-1}d\, C)^T s+  \tfrac{1}{2}s^T (C^{-1}d\, C) s \\
& =- \vec(C^{-T})^T d\,\vec(C) + \vec(C^{-T}ss^T)^T d\, \vec(C) \\
&=  \vec(C^{-T}ss^T - CT{-T})^T E^T d\, \vech(C). \\
\therefore \nabla_{\vech(C)} \log q_\lambda (\theta)  &= E \vec(C^{-T}ss^T - C^{-T}) = \vech(C^{-T}ss^T - C^{-T}).
\end{aligned}
\end{equation*}
Here $E$ denotes the $p \times p$ elimination matrix \citep{Magnus1980}, which has the following properties, (i) $E\vec(A) = \vech(A)$ for any $p \times p$ matrix $A$ and (ii) $E^T \vech(A) = \vec(A)$ if $A$ is a $p \times p$ lower triangular matrix of order $p$.

\section{Laplace approximation}
We compare the variational methods with Laplace approximation which also approximates the posterior distribution using a Gaussian density.  Consider a second order Taylor approximation to $L(\theta) = \log p(y, \theta)$ at the posterior mode, $\hat{\theta}^* = \argmax_\theta L(\theta)$. We have
\begin{equation*}
\begin{aligned}
L(\theta) \approx L(\hat{\theta}^*) + \frac{1}{2} (\theta - \hat{\theta}^*)^T \nabla^2_\theta L(\hat{\theta}^*) (\theta - \hat{\theta}^*)
\end{aligned}
\end{equation*}
since $\nabla_\theta L(\hat{\theta}^*) = 0$ at the mode. Thus 
\begin{equation*}
\begin{aligned}
p(\theta|y) &\propto p(\theta, y) \\
&\mathrel{\dot\propto} \exp \left\{ \frac{1}{2} (\theta - \hat{\theta}^*)^T \nabla^2_\theta L(\hat{\theta}^*) (\theta - \hat{\theta}^*) \right\}.
\end{aligned}
\end{equation*}
Thus $p(\theta|y)$ can be approximated by $N(\hat{\theta}^*, - \{\nabla^2_\theta L(\hat{\theta}^*)\}^{-1})$. We use the adjusted pseudolikelihood $\tilde{f}(y|\theta)$ as a plug-in for the true likelihood $p(y|\theta)$. Let $\tilde{L}(\theta) = \log \tilde{p}(y, \theta)$ which is given in \eqref{logjoint}. Then 
\begin{equation*}
\begin{aligned}
d \tilde{L}(\theta) &= \sum_{(i,j) \in \D} \{ y_{ij}   -b'[ \delta_s(y)_{ij}^Tg(\theta) ] \}\delta_s(y)_{ij}^TWd\theta - (\theta -  \mu_0)^T \Sigma_0^{-1}d\theta. \\
\implies \nabla_\theta \tilde{L}(\theta) &= W^T \sum_{(i,j) \in \D} \{ y_{ij} - b'[ \delta_s(y)_{ij}^Tg(\theta) ] \} \delta_s(y)_{ij}  - \Sigma_0^{-1}( \theta - \mu_0).
\end{aligned}
\end{equation*}
\begin{equation*}
\begin{aligned}
d \nabla_\theta \tilde{L}(\theta) &= - W^T \sum_{(i,j) \in \D}  b''[ \delta_s(y)_{ij}^Tg(\theta) ] \delta_s(y)_{ij} \delta_s(y)_{ij}^TWd\theta - \Sigma_0^{-1} d\theta. \\
\implies \nabla_\theta^2 \tilde{L}(\theta) &= - W^T \sum_{(i,j) \in \D}  b''[ \delta_s(y)_{ij}^Tg(\theta) ] \delta_s(y)_{ij} \delta_s(y)_{ij}^T W- \Sigma_0^{-1}.
\end{aligned}
\end{equation*}
We find an estimate of the posterior mode by finding the zero of $\nabla_\theta \tilde{L}(\theta)$ numerically using the L-BFGS Algorithm via the {\tt optimize} function in the Julia package {\tt Optim}.

\bibliographystyle{apalike}
\bibliography{ref}

\end{document}